\documentclass[Journal]{IEEEtran}

\usepackage{amsmath}
\usepackage{amssymb}
\usepackage{mathtools}
\usepackage{graphicx}
\usepackage[linesnumbered,ruled]{algorithm2e}
\usepackage[usenames,dvipsnames]{color}
\usepackage{xcolor}
\usepackage{float}
\usepackage[english]{babel}
\usepackage{bm}
\usepackage{multirow}
\usepackage{tikz}
\usepackage{arydshln}
\usetikzlibrary{matrix}
\usepackage[font={small,it}]{caption}
\usepackage{subcaption}
\usepackage[flushleft]{threeparttable}
\newcommand{\RNum}[1]{\uppercase\expandafter{\romannumeral #1\relax}}
\newtheorem{theorem}{Theorem}[section]

\newtheorem{lemma}[theorem]{Lemma}
\newtheorem{definition}[theorem]{Definition}
\newtheorem{proposition}[theorem]{Proposition}
\newtheorem{remark}{Remark}
\newtheorem{assumption}[theorem]{Assumption}
\newtheorem{example}{Example}

\newtheorem{innercustomthm}{Theorem}
\newenvironment{customthm}[1]
  {\renewcommand\theinnercustomthm{#1}\innercustomthm}
  {\endinnercustomthm}
  
\newtheorem{innercustompro}{Proposition}
\newenvironment{custompro}[1]
  {\renewcommand\theinnercustompro{#1}\innercustompro}
  {\endinnercustompro}
  
\newtheorem{innercustomlem}{Lemma}
\newenvironment{customlem}[1]
  {\renewcommand\theinnercustomlem{#1}\innercustomlem}
  {\endinnercustomlem}

\begin{document}

\title{Optimal Control of Large-Scale Networks using Clustering Based Projections}

\author{Nan~Xue,~\IEEEmembership{Student Member,~IEEE},
        Aranya~Chakrabortty,~\IEEEmembership{Senior Member,~IEEE}

\thanks{N. Xue and A. Chakrabortty are with the Department
of Electrical and Computer Engineering, North Carolina State University, Raleigh,
NC, 27695 USA, e-mail: nxue@ncsu.edu, achakra2@ncsu.edu }
\thanks{The work is supported partly by the US National Science Foundation (NSF) under grant ECCS 1054394.}
}

\maketitle

\begin{abstract}
In this paper we present a set of projection-based designs for constructing simplified linear quadratic regulator (LQR) controllers for large-scale network systems. When such systems have tens of thousands of states, the design of conventional LQR controllers becomes numerically challenging, and their implementation requires a large number of communication links. Our proposed algorithms bypass these difficulties by clustering the system states using structural properties of its closed-loop transfer matrix. The assignment of clusters is defined through a structured projection matrix $P$, which leads to a significantly lower-dimensional LQR design. The reduced-order controller is finally projected back to the original coordinates via an inverse projection. The problem is, therefore, posed as a model matching problem of finding the optimal set of clusters or $P$ that minimizes the $\mathcal{H}_{2}$-norm of the error between the transfer matrix of the full-order network with the full-order LQR and that with the projected LQR. We derive a tractable relaxation for this model matching problem, and design a $P$ that solves the relaxation. The design is shown to be implementable by a convenient, hierarchical two-layer control architecture, requiring far less number of communication links than full-order LQR. 
\end{abstract}


\IEEEpeerreviewmaketitle

\section{Introduction}
A vast majority of practical networked dynamic systems (NDS), ranging from power system networks to wireless networks to social or biological networks, consist of several hundreds to thousands of nodes that are spatially distributed over wide geographical spans. Developing tractable control designs for such large complex networks, and implementing those designs through affordable communication, continue to be a challenge for network designers. Traditionally, control theorists have addressed the problem of controlling large-dimensional systems by imposing structure on controllers. The most promising approach, for example, started with the idea of decentralized control \cite{siljak}, followed by techniques such as singular perturbation theory \cite{singular,chow}, balanced truncation \cite{balanced,wreduct,antoulas}, and $\nu$-gap reduction \cite{rc} among others. These methods aim to simplify the design of controllers for large systems by exploiting weak coupling between their state variables, and by ignoring states that are `less important' than others. The trade-off, however, is that the resulting controllers are often agnostic of the natural coupling between the states, especially the coupling between the closed-loop states, since many of these couplings are forcibly eliminated to facilitate the design itself. Therefore, extending these methods to facilitate controller designs for networks, especially to NDS whose states may be defined over highly structured topologies such as clustering, is quite difficult. The literature for developing tangible and yet simple low-dimensional controllers that satisfy global stability and dynamic performance requirements of very large NDS is still unfortunately sparse. Ideas on aggregate control \cite{boker}, {\it glocal} control \cite{glocal} and hierarchical control \cite{imuracdc}, \cite{Madji} have recently been proposed to address this challenge. The goal of these designs, however, is to guarantee stability by modular tuning of local controller gains; their degrees of freedom for guaranteeing performance can be limited.

To bridge this gap, in this paper we propose a design method called {\it control inversion}. The approach is to cluster the states of an $n$-dimensional network into $r>0$ distinct, non-overlapping groups. We assume $n$ to be a large positive integer, and $r\leq n$ to be a given design parameter. The grouping is defined by a $(r\times n)$ structured projection matrix $P$ whose elements denote the identity of states in the clusters, weighted by certain projection weights. The design thereafter consists of three steps. First, for the full-order network an $n$-dimensional LQR controller is defined for any given choice of $Q$ and $R$. We refer to this controller as the {\it benchmark} LQR. Second, the projection matrix $P$ is used to construct an $r$-dimensional reference model for which an $r$-dimensional LQR controller is designed. The design matrices for this reduced-order controller, however, are not free; they are constrained by being related to $Q$ and $R$ through $P$. The important point, however, is that the design dimension reduces to $r$ from $n$. Finally, this reduced-order controller is projected back to the full-order network by the inverse projection $P^{T}$. The problem is then to find a projection matrix $P$ that minimizes the $\mathcal{H}_{2}$-norm of the error between the transfer function matrices of the full-order network with the benchmark LQR controller and with the projected LQR controller.

This problem by itself, however, is non-convex even without any structural constraint on $P$. To bypass this intractability, we introduce a relaxation, which is done in two stages. The first stage relaxes the error minimization to the minimization of its upper bound, while the second stage applies a low-rank approximation. We finally design a $P$ that solves this relaxed problem. Three distinct variants of the design are proposed. In the first case, we optimize over cluster assignment while keeping the projection weights fixed, and establish that this minimization can be posed as an unsupervised clustering problem. We use weighted k-means \cite{kmeans} to solve this minimization. In the second case, we fix the cluster identities, and optimize over the projection weights. Depending on the stability of the open-loop system, we show that this minimization can be posed as finding the dominant eigenvector of the controllability Gramian or as finding the $Z$-eigenvector of a tensor \cite{tensor1}. In the third case, we propose an iterative method to optimize over both cluster assignment and projection weights. The controllers resulting from all three algorithms are shown to be implementable by a convenient, hierarchical two-layer control architecture, requiring far less number of communication links than full-order LQR as well as sparsity-promoting LQR \cite{sparse}. 

Recently, \cite{groningen} and \cite{h2} have used structural projection-based ideas for model reduction of large networks, but not for control designs. Attention has also been drawn to designing LQR controllers for large systems by finding low-rank solutions of algebraic Riccati equations \cite{are}. However, like most Krylov subspace-based reduction methods such as in \cite{antoulas}, the controller in \cite{are} is unstructured, and hence demands as many communication links as the full-order LQR itself. Distributed controllers using model matching \cite{qinv}, sparsity-promoting LQR in \cite{sparse} and structured LQR in \cite{struct,Fattahi} promise to reduce the communication density, but their designs inherit the same dimensionality as the full-order design. Unlike all of these methods, the novelty of our algorithms is in the facilitation of closed-loop control from the perspective of both design and implementation. The recent papers \cite{local1,local2} also address both goals, but the dimensionality of their controllers is subject to the sparsity structure of the open-loop network while our design does not necessarily require any such sparsity. Preliminary results on this design have been presented in the recent conference paper \cite{cdc16}, but only for a consensus model with specific $Q$ and $R$ matrices.  

The remainder of the paper is organized as follows. Section \RNum{2} formulates the problem of clustering-based optimal control. The relaxation for the original problem is derived in Section \RNum{3}, which is then solved by a clustering algorithm based on weighted k-means optimization in Section \RNum{4}. The design for cluster weights as well as that for the weights and clusters taken together are discussed in Section \RNum{5}. All three algorithms are illustrated via simulations in Section \RNum{6}. Section \RNum{7} concludes the paper.

{\bf{Notation} } We will use the following notations throughout this paper:%

{
\renewcommand{\arraystretch}{1.1}
\begin{tabular}{lp{5.8cm}}
$|m|$ & absolute value of a scalar $m$ \\
$|\mathcal{S}|_{c}$ & cardinality of a set $\mathcal{S}$  \\
$\bm{1}_{n}$ & column vector of size $n$ with all $1$ entries \\
$I_{k}$ & identity matrix of size $k$ \\
$M_{i,j}$ & the $(i,j)^{th}$ entry of a matrix $M$ \\
$diag(m)$ & diagonal matrix with vector $m$ on its principal diagonal \\
$M\otimes N$ & Kronecker product of $M$ and $N$\\
$M\circ N$ & Hadamard product of $M$ and $N$\\
$tr(M)$ & trace operation on a matrix $M$ \\
$\| M \|_{F}$ & Frobenius norm of a matrix $M$, i.e. $\| M \|_{F}=\sqrt{tr(MM^{T})}$ \\
$ker(M)$ & kernel of a matrix $M$  \\
$\bar{\sigma}(M)$, $\bar{\lambda}(M)$ & largest singular value, or eigenvalue with largest real part of a matrix $M$ \\
\end{tabular}
\begin{tabular}{lp{5.8cm}}
$\underline{\sigma}(M)$, $\underline{\lambda}(M)$ & smallest singular value, or eigenvalue with smallest real part of a matrix $M$ \\
$\bar{v}(M)$ & right eigenvector of $\bar{\lambda}(M)$ \\
\end{tabular} }%

Given a matrix $M=[m_{1},...,m_{n}]\in \mathbb{R}^{n\times n}$, its vector form is defined by $vec(M)=[m_{1}^{T},...,m_{n}^{T}]^{T}$, with the inverse operation defined by $unvec(vec(M))=M$. A transfer matrix is defined as $g(s)=C(sI-A)^{-1}B+D$, with a realization form of $g(s)=\left[
\begin{array}{c|c}
A & B \\ \hline
C & D
\end{array}
\right]
$. We refer to $g(s)$ as stable if $A$ is Hurwitz, and unstable otherwise. Furthermore, the $\mathcal{H}_{2}$ and $\mathcal{H}_{\infty}$ norms of a stable transfer matrix $g(s)$ are defined by $\|g(s)\|_{\mathcal{H}_{2}}=\sqrt{\int_{-\infty}^{\infty}tr[g^{*}(t)g(t)]\mathrm{d}t}=\sqrt{\frac{1}{2\pi}\int_{-\infty}^{\infty}tr[g^{*}(j\omega)g(j\omega)]\mathrm{d}\omega}$ and $\|g(s)\|_{\mathcal{H}_{\infty}}=sup_{\omega}\ \bar{\sigma}[ g(j\omega)] $. 

From graph theory, a graph $\mathcal{G}=(\mathcal{V},\mathcal{E})$ is defined over a node (vertex) set $\mathcal{V}=\{ 1,...,n \}$ and an edge set $\mathcal{E} \subset \mathcal{V} \times \mathcal{V}$, which contains two-element subsets of $\mathcal{V}$. If $\{i,j\} \in \mathcal{E}$, we call nodes $i$ and $j$ adjacent, and denote the relation by $i \sim j$, or simply $i j$. The set of nodes adjacent to $i \in \mathcal{V}$ is noted by $\mathcal{N}_{i} = \{ j\in \mathcal{V} | i \sim j \}$. In this paper, $\mathcal{G}$ is assumed to be undirected, which implies $ij$ is equivalent to $ji$, and there are no loops or multiple edges between nodes.

\section{Problem Formulation}

Consider a general LTI system of the form
\begin{align}
\begin{cases}
\dot{x}(t) = Ax(t) + Bu(t) + B_{d}d(t),\quad x(0) = x_{0} \\
y(t) = Cx(t)
\end{cases},
\label{full}
\end{align}
where $x(t) \in \mathbb{R}^{n}$, $u(t) \in \mathbb{R}^{m}$ and $y(t) \in \mathbb{R}^{p}$ represent the vector of state, control and output variables respectively, and $d(t)\in \mathbb{R}^{n_{b}}$ is a disturbance entering into the system. We assume (\ref{full}) to be defined over a network of $n_{s}\leq n$ interconnected subsystems, with their network topology represented by a connected graph $\mathcal{G} = (\mathcal{V},\mathcal{E})$, $\mathcal{V} = \{ 1,...,n_{s} \}$. The dynamics of each subsystem can be written as
\begin{align}
\begin{cases}
\dot{x}_{i}(t) {=} A_{ii}x_{i}(t) {+} \underset{j\in \mathcal{N}_{i}}{\sum} A_{ij}x_{j}(t) {+} B_{i}u_{i}(t) {+} B_{di}d_{i}(t) \\
y_{i}(t) = C_{i}x_{i}(t)
\end{cases}
\label{nodesys}
\end{align}
$i = 1,...,n_{s}$, where $A_{ij}$, $B_{i}$, $B_{di}$ and $C_{i}$ are submatrices with compatible dimensions from $A$, $B=diag(B_{1},...,B_{n_{s}})$, $B_{d}=diag(B_{d1},...,B_{dn_{s}})$ and $C=diag(C_{1},...,C_{n_{s}})$. Notice that the dimension of $u_{i}(t)$ can be zero, meaning that the $i^{th}$ subsystem can have no input. 

In this paper, we consider an LQR design for (\ref{full}), and assume $C=I_{n}$ for full-state feedback. Given two real-valued matrices $Q=Q^{T} \succeq 0$ and $R = R^{T} \succ 0$, the LQR problem is posed as finding a feedback law $u(t)=-Kx(t)$ such that the cost function
\begin{align}
J := \int_{0}^{\infty} [ x^{T}(t)Qx(t)+u^{T}(t)Ru(t)] \mathrm{d}t 
\label{LQR}
\end{align}
is minimized. The expression (\ref{LQR}), also known as the infinite-horizon continuous-time LQR, can be solved by the following algebraic Riccati equation (ARE)
\begin{align}
A^{T}X + XA + Q - XGX=0, \label{are}
\end{align}
where $G=BR^{-1}B^{T}$. The feedback matrix can be found through $K=R^{-1}B^{T}X$. For such a solution $X$ to exist, we will adhere to the following assumption throughout this paper.
\begin{assumption}
$(Q^{\frac{T}{2}},A)$ is observable, and $(A,BR^{-\frac{1}{2}})$ is stabilizable.
\label{asss}
\end{assumption}

According to \cite{rc}, the assumption above guarantees a unique stabilizing solution $X=X^{T}\succ 0$. However, finding this solution from (\ref{are}) in practice is subject to $\mathcal{O}(n^{3})$ computational complexity, which can become unscalable for large-scale systems. Moreover, the resulting matrix $X$ is usually an unstructured dense matrix, which demands every subsystem in the network to communicate with every other subsystem for implementing the feedback. These two factors together make both the design and implementation of $u=Kx$ very difficult, especially when $\mathcal{G}$ consists of thousands to tens of thousands of nodes. Therefore, we propose a design strategy, which we refer to as {\it control inversion}, to repose this LQR problem using a clustering-based projection. 

\subsection{Control Inversion}

\begin{definition}
\label{Pd}
Given an integer $r$, where $0<r\leq n$, and a non-zero vector $w \in \mathbb{R}^{n}$, define $r$ non-empty, distinct, and non-overlapping subsets of the state index set $\mathcal{V}_{s} = \{ 1,...,n \}$, respectively denoted as $\mathcal{I} = \{ \mathcal{I}_{1},...,\mathcal{I}_{r} \}$, such that $\mathcal{I}_{1} \cup ... \cup \mathcal{I}_{r} = \mathcal{V}_{s}$. A clustering-based projection matrix $P \in \mathbb{R}^{r\times n}$ is defined as
\begin{align}
P_{i,j} := \begin{cases} 
\frac{w_{j}}{\| w_{\mathcal{I}_{i}} \|_{2}} & \quad  j \in \mathcal{I}_{i} \\ 
0 & \quad \text{otherwise} 
\end{cases},
\label{eqp}
\end{align}
where $w_{\mathcal{I}_{i}}=[w_{\mathcal{I}_{i}\{1\}},...,w_{\mathcal{I}_{i}\{ |\mathcal{I}_{i} |_{c}\}}]^{T}$ is non-zero, and $\mathcal{I}_{i}\{ j\}$ denotes the $j^{th}$ element in the set $\mathcal{I}_{i}$. The matrix $P$ has the following three properties: 
\begin{itemize}
\item It is row orthonormal, i.e. $PP^{T}=I_{r}$;
\item Image of $P^{T}P$ lies in the span of $w$, i.e. $P^{T}Pw=w$;
\item Given $v\neq 0$, $Pv = 0$ only if $w_{\mathcal{I}_{i}}^{T}v_{\mathcal{I}_{i}} = 0$, $i=1,...,r$.
\end{itemize}
\end{definition}

The construction of $P$ is shown by the following example.
\begin{example}
Let $w = \begin{bmatrix}
1 & 1 & 1 & 2 & 1 & 1 & 1 & 1 & 1 & 1
\end{bmatrix}^{T}$, $\mathcal{I}_{1}=\{ 1,2\}$, $\mathcal{I}_{2}=\{ 3,4,5\}$ and $\mathcal{I}_{3}=\{ 6,7,8,9,10 \}$. Then,
\begin{align*}
P{=}\begin{bmatrix}
\frac{1}{\sqrt{2}} & \frac{1}{\sqrt{2}} & 0 & 0 & 0 & 0 & 0 & 0 & 0 & 0 \\
0 & 0 & \frac{1}{\sqrt{6}} & \frac{2}{\sqrt{6}} & \frac{1}{\sqrt{6}} & 0 & 0 & 0 & 0 & 0\\
0 & 0 & 0 & 0 & 0 & \frac{1}{\sqrt{5}} & \frac{1}{\sqrt{5}} & \frac{1}{\sqrt{5}} & \frac{1}{\sqrt{5}} & \frac{1}{\sqrt{5}}
\end{bmatrix}.
\end{align*}
\end{example} 

Given $P$ defined over any clustering set $\mathcal{I}$ and weight vector $w$, the control inversion strategy for the LQR problem (\ref{LQR}) is then composed of the following three steps.

\subsubsection{Projection to reduced-order system}
Using projection $P$, we first construct a reduced-order model
\begin{align}
\begin{cases}
\dot{\tilde{x}}(t) & = \tilde{A}\tilde{x}(t)+\tilde{B}\tilde{u}(t)+\tilde{B}_{d}d(t),\quad \tilde{x}(0) = \tilde{x}_{0} \\
\tilde{y}(t) & =\tilde{C}\tilde{x}(t) 
\end{cases},
\label{reduced} 
\end{align}
where $\tilde{A} :=PAP^{T}\in \mathbb{R}^{r\times r}$, $\tilde{B} :=PB \in \mathbb{R}^{r\times m}$, $\tilde{B}_{d} :=PB_{d}\in \mathbb{R}^{r\times n_{b}}$ and $\tilde{C} :=PCP^{T}=I_{r}$. For this system $\tilde{x} \in \mathbb{R}^{r}$ is the state and $\tilde{u}\in \mathbb{R}^{m}$ is the control input.

\subsubsection{Reduced-order LQR design} 
We similarly project the LQR parameters by $\tilde{Q} = PQP^{T} \in \mathbb{R}^{r \times r}$, and let $\tilde{R} = R$ such that $\tilde{G} := \tilde{B}\tilde{R}^{-1}\tilde{B} = PGP^{T} \in \mathbb{R}^{r \times r}$. An LQR problem for the reduced-order model (\ref{reduced}) is then posed as to minimize
\begin{align}
\tilde{J} := \int_{0}^{\infty} [ \tilde{x}^{T}(t)\tilde{Q}\tilde{x}(t)+\tilde{u}^{T}(t)\tilde{R}\tilde{u}(t)] \mathrm{d}t
\label{LQRreduced} 
\end{align}
with respect to $\tilde{u}(t) = - \tilde{K} \tilde{x}(t)$. Here the feedback matrix $\tilde{K} = \tilde{R}^{-1}\tilde{B}^{T}\tilde{X}$ corresponds to the solution $\tilde{X} \in \mathbb{R}^{r \times r}$ of the reduced-order ARE of (\ref{LQRreduced}), which is written as 
\begin{align}
\tilde{A}^{T}\tilde{X} + \tilde{X}\tilde{A} + \tilde{Q} - \tilde{X}\tilde{G}\tilde{X}=0. \label{reare} 
\end{align}

\subsubsection{Inverse projection to original coordinates} 
The solution $\tilde{X}$ from (\ref{reare}) is projected back to the original coordinates through the inverse projection
\begin{align}
\hat{X} = P^{T}\tilde{X}P.
\label{invert} 
\end{align}
This projected controller can then be implemented in the full-order model (\ref{full}) using $u=-R^{-1}B^{T}\hat{X}x$, which implies that the effective feedback gain matrix is 
\begin{align}
\hat{K} =R^{-1}B^{T}\hat{X}. \label{Khat}
\end{align} 

\subsection{Problem Statement}

The controller $\hat{K}$ is dependent on the projection $P$ through equations (\ref{reduced}), (\ref{LQRreduced}), and (\ref{invert}). The choice of $P$ is guided in the following way. Consider 
\begin{align}
g(s) := (sI_{n} - A + BK)^{-1}B_{d},
\label{fullclptf}
\end{align}
which is the closed-loop transfer matrix from $d$ to $x$ for (\ref{full}) with full-order LQR. Similarly, consider 
\begin{align}
\hat{g}(s) &:= (sI_{n}-A+B\hat{K})^{-1}B_{d},
\label{clptff}
\end{align}
which is the closed-loop transfer matrix from $d$ to $x$ for (\ref{full}) with the projected controller (\ref{Khat}). Using (\ref{fullclptf}) and (\ref{clptff}), we next state our main problem of interest.

{\bf Main problem:} Given system (\ref{full}) and an integer $r>0$, the problem addressed in this paper is to find a clustering set $\mathcal{I}$ and a non-zero vector $w$ such that the corresponding projection matrix $P$ solves the model matching problem\footnote{The initial condition for the reduced-order model (\ref{reduced}) does not need to be related to that of the full-order model (\ref{full}). Our goal is to compare (\ref{fullclptf}) and (\ref{clptff}), both of which have zero initial conditions.} 
\begin{equation}
\begin{aligned}
& \underset{P}{\mathrm{minimize}}
& &  \| g(s)-\hat{g}(s) \|_{\mathcal{H}_{2}} .
\end{aligned}
\label{MM} 
\end{equation} 

However, finding an exact solution for this optimization problem is intractable given that the objective function is an implicit and non-convex function of $P$, and also because $P$ is defined over a combinatorial structure. Our main contribution, therefore, is finding a tractable relaxation for (\ref{MM}) as a quadratic function of $P$, and thereafter designing $P$ to solve the relaxed problem.  We make the following assumption so that $g(s)$ and $\hat{g}(s)$ both have minimal realization.
\begin{assumption}
The pair $(A,B_{d})$ is controllable.
\end{assumption}

{\bf Solution Strategy:} The outline of our solution strategy is as follows. In Section \RNum{3}, we derive an upper bound relaxation for (\ref{MM}) such that its objective function is quadratic in $P$. Ideally speaking, one can solve for $P$ from this relaxed problem. The computational complexity for constructing the objective function is, however, $\mathcal{O}(n^{3})$ since it requires the computation of the controllability Gramian of $g(s)$. To bypass this difficulty, a second round of relaxation is applied by exploiting the low-rank (denoted as $\kappa$) structure of the controllability Gramian. After these two relaxations, the final objective function, still quadratic in $P$, can be constructed in $\mathcal{O}(n\kappa^{2})$ complexity, which is near linear if $\kappa \ll n$. The solution to this optimization is then addressed in two ways - first by finding the clustering set $\mathcal{I}$ with a fixed weight vector $w$ (Section \RNum{4}), and second, by finding $w$ while keeping $\mathcal{I}$ fixed (Section \RNum{5}). We also propose to combine these two approaches by an iterative algorithm. The overall design flow and the numerical complexities for each step are previewed in Fig. \ref{compdiag}. Detailed explanations of these complexities will be provided in the respective sections to follow.

\begin{figure}[H]
\centering
\includegraphics[width=1\columnwidth]{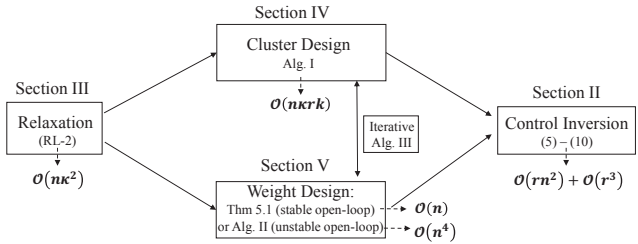}
\caption{Step-by-step execution of the proposed designs}
\label{compdiag}
\end{figure}

\subsection{Benefits of Control Inversion}

\begin{figure*}
    \centering
    \begin{subfigure}[t]{0.3\textwidth}
    \centering
        \includegraphics[width=0.9\columnwidth]{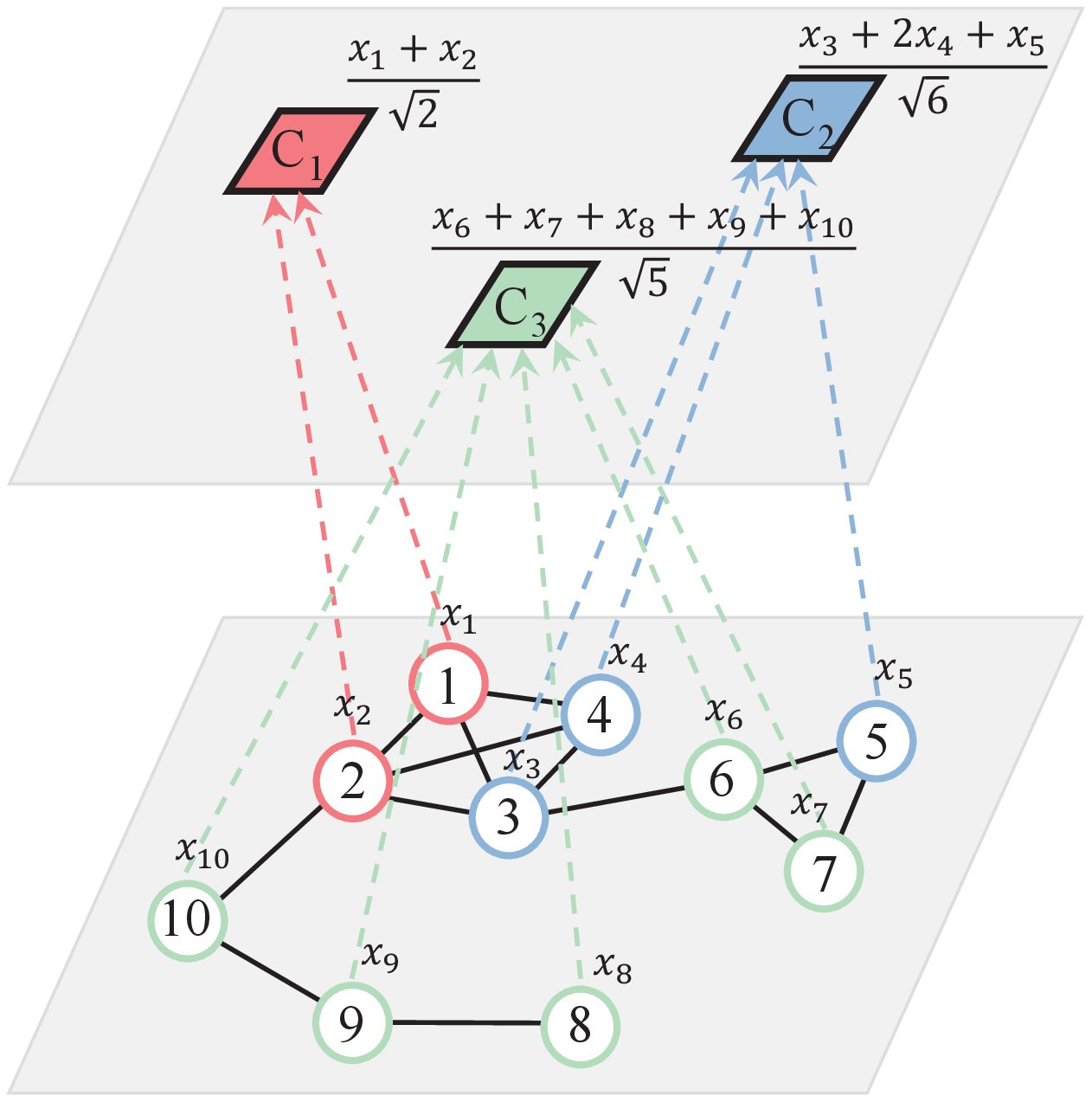}
        \caption{Step 1 - state averaging $Px$}
        \label{cp1}
    \end{subfigure}
    ~        
    \begin{subfigure}[t]{0.3\textwidth}
    \centering
        \includegraphics[width=0.9\columnwidth]{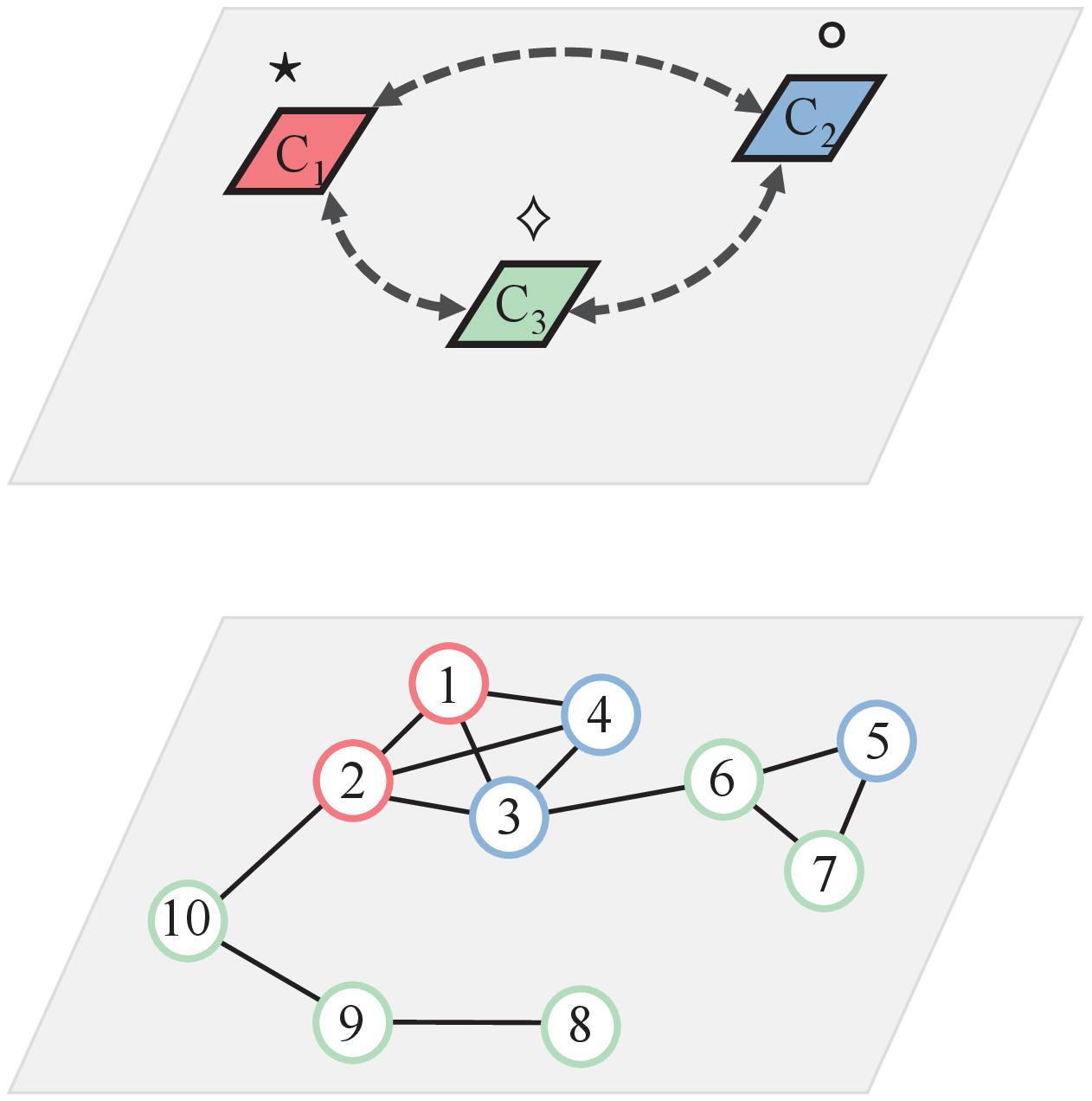}
        \caption{Step 2 - reduced-order control $\tilde{X}Px= [\star \ \circ \ \diamond ]^{T}$}
        \label{cp2}
    \end{subfigure}
    ~        
    \begin{subfigure}[t]{0.3\textwidth}
    \centering
        \includegraphics[width=0.9\columnwidth]{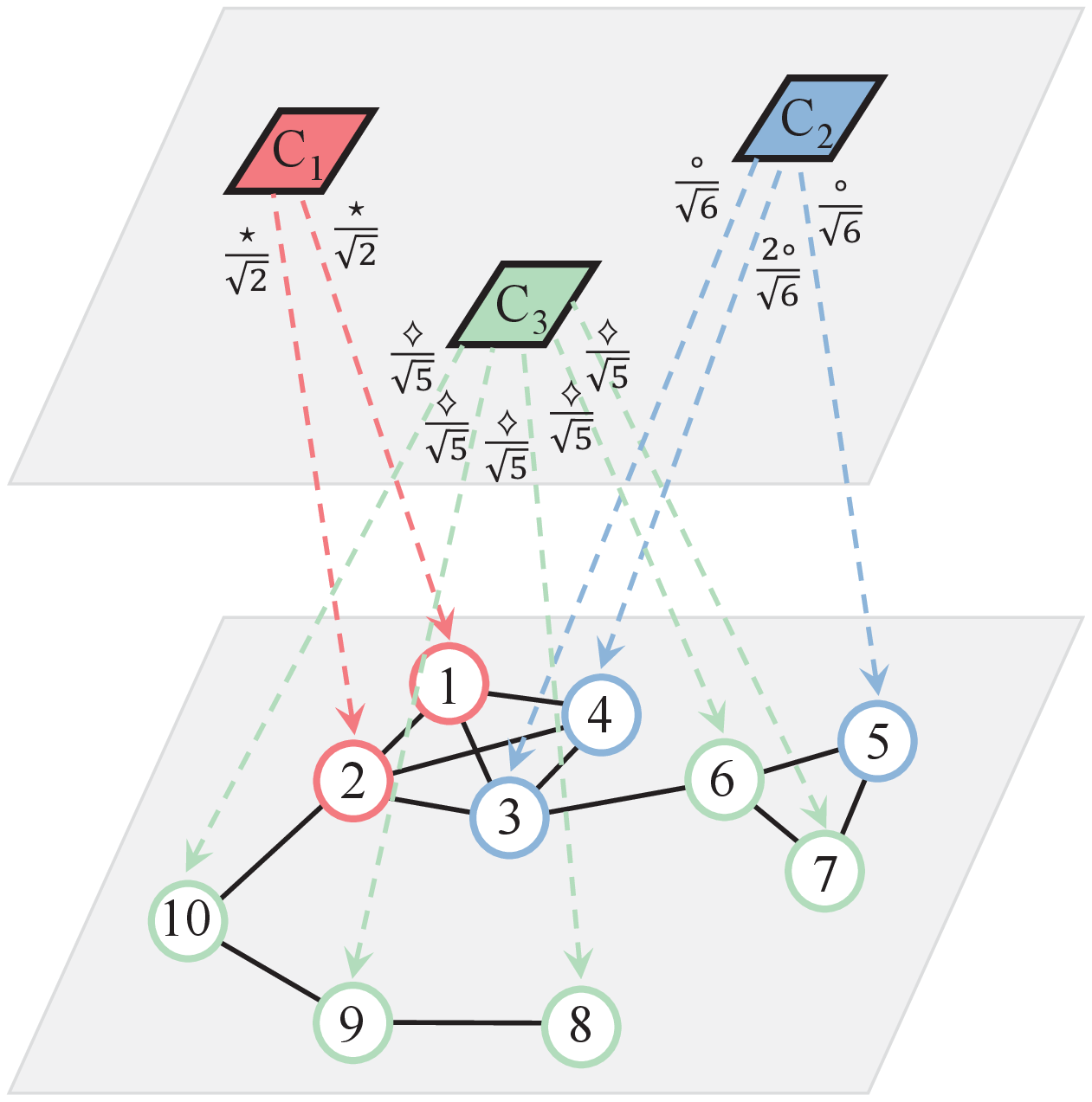}
        \caption{Step 3 - control inversion $u = -\hat{X}x$}
        \label{cp3}
    \end{subfigure}
    \caption{Cyber-physical architecture for implementing the feedback controller $\hat{K}$. Here $\mathcal{I}_{1} = \{ 1,2 \}$, $\mathcal{I}_{2} = \{ 3,4,5 \}$ and $\mathcal{I}_{3} = \{ 6,7,8,9,10 \}$. Solid lines represent physical connections, while dashed lines represent communication links. $C_{1}$, $C_{2}$, $C_{3}$ are coordinators for clusters $1$, $2$ and $3$. For the simplicity of illustration, for this example we assume $R^{-1}B^{T}=I_{n}$.}\label{cp}
\vspace{-1em}
\end{figure*}

An important point to note is that the physical meaning of the state $\tilde{x}(t)$ of the reduced-order model (\ref{reduced}) has no relation to that of the state $x(t)$ of our full-order model (\ref{full}). This is a key difference of the control inversion design from traditional model-reduction based designs where the reduced-order state vector is typically a direct projection of the full-order state vector. The projection in our design is rather applied on the controller $\tilde{X}$ instead of $x(t)$. Two natural benefits of this approach are as follow:

{1. \it Tractability of design:}  The computational complexity for constructing the reduced-order ARE in (\ref{reare}) is $\mathcal{O}(rn^{2})$, while that for solving this ARE is $\mathcal{O}(r^{3})$. The computational complexity required to design $P$ through Sections \RNum{3}, \RNum{4} and \RNum{5} will be shown to be simpler than the $\mathcal{O}(n^{3})$ complexity of a full-order LQR design. Thus, if $r\ll n$, the overall control inversion design becomes numerically more tractable than full-order LQR.

{2. \it Simplicity in implementation:} The projected matrix $\hat{X}=P^{T}\tilde{X}P$ is a structured $r$-ranked matrix, which results in a sequential two-layer hierarchical control architecture. The implementation of the feedback $u(t) =-R^{-1}B^{T}\hat{X}x(t)$ follows three steps. First, a coordinator is assigned to each cluster $\mathcal{I}_{i}$, which collects the measurements of all the states belonging to that cluster. Each coordinator then computes the weighted averaged state $P_{i,:}x$ for its cluster, $i=1,2,..,r$. Next, the coordinators exchange these weighted averages, and each of them compute the $r$-dimensional vector $\tilde{X}Px$. Note that in this process no coordinator will be able to infer individual state measurements from other clusters. Finally, the coordinator of $\mathcal{I}_{i}$ computes the control vector $u_{i}$ by taking linear combinations of the elements of $\tilde{X}Px$. The linear combination follows from $u = -R^{-1}B^{T}P^{T}\tilde{X}Px$. The individual elements of $u_{i}$ are broadcast to the respective input actuators in $\mathcal{I}_{i}$. Note that since the aggregation is applied to the state $x(t)$, and not on the subsystems, the state vector of any subsystem can be partitioned among different clusters. Thus, in practice, a subsystem may need to transmit its states to more than one coordinator, and also receive control inputs from more than one coordinator. We illustrate the three implementation steps by an example in Fig. \ref{cp}.

In the worst-case scenario when every subsystem has a scalar state $x_{i}$ and a scalar control input $u_{i}$, the two-layer control implementation will result in a much sparser communication topology with $n+{{r}\choose{2} }$ bidirectional links compared to an LQR controller which would require ${{n}\choose{2}}$ number of links. This reduction, combined with standard networking protocols such as multi-casting \cite{alhemy}, makes the implementation of our proposed controller convenient and cheap.

\section{Relaxations for Model Matching}

In this section we describe the theoretical derivation of the relaxation for the model matching problem (\ref{MM}). We start by discussing the well-posedness of the projected controller $\hat{K}$ in Section \RNum{3}.A. The final optimization to be solved is then obtained from two stages of relaxations as detailed in Section \RNum{3}.B and \RNum{3}.C respectively. All proofs are presented in the Appendix B.

\subsection{Well-Posedness Conditions}

The well-posedness of $\hat{K}$ is equivalent to two conditions - namely, if the reduced-order ARE (\ref{reare}) admits a solution $\tilde{X}$, and if $\hat{g}(s)$ is stable. We discuss these two factors as follows.

\subsubsection{Existence condition}

Similar to the full-order ARE (\ref{are}), the reduced-order ARE (\ref{reare}) is guaranteed with a unique solution $\tilde{X}\succ 0$ if $(\tilde{A},\tilde{G}^{\frac{1}{2}})$ is stabilizable and $(\tilde{Q}^{\frac{T}{2}},\tilde{A})$ is observable \cite{rc}. However, unlike conventional model reduction techniques that can utilize unstructured projections to preserve the exact stabilizability and observability properties from $(A,G^{\frac{1}{2}})$ and $(Q^{\frac{T}{2}},A)$, the structured projection $P$ in (\ref{eqp}) does not guarantee that. As a result, for a general system the ARE (\ref{reare}) may not admit a unique solution $\tilde{X}$, and therefore $\hat{K}$ may not exist. To bypass this problem, we modify the definitions of $\tilde{Q}$ and $\tilde{G}$ by using a constant shift. That is, in case the pair $(PAP^{T},PG^{\frac{1}{2}})$ is not stabilizable and/or $(Q^{\frac{T}{2}}P^{T},PAP^{T})$ is not observable, we let 
\begin{align}
\tilde{Q} = PQP^{T} + \alpha I_{r} \quad \text{and/or} \quad \tilde{G} = PGP^{T} + \alpha I_{r} \label{QRmod}
\end{align}
for a small constant shift $\alpha >0$. Using (\ref{QRmod}), $\tilde{Q}$ and $\tilde{G}$ become positive-definite matrices, which yield $(\tilde{A},\tilde{G}^{\frac{1}{2}})$ controllable and/or $(\tilde{Q}^{\frac{T}{2}},\tilde{A})$ observable. The existence of $\tilde{X}$ and $\hat{K}$ can thus be guaranteed. The matching error introduced by this shift will be discussed in the next subsection.

\subsubsection{Bound of ARE solution}
We first recall a lemma from \cite{arebound} about the eigenvalue bound of a general ARE solution.
\begin{lemma}
(Lemma 1.5 in \cite{arebound})
The solution $X$ from the ARE (\ref{are}) satisfies the following upper bound:
\begin{align}
\bar{\lambda}(X) \leq \beta(A,G,Q) = \bar{\lambda}(D_{t})\frac{\bar{\lambda}[(Q+K_{t}^{T}K_{t})D_{t}]}{\underline{\lambda}(FD_{t})}, 
\end{align}
where $K_{t}$ is any matrix stabilizing $A + G^{\frac{1}{2}}K_{t}$, and $D_{t}$ and $F$ are positive-definite matrices that satisfy 
\begin{align}
(A{+}G^{\frac{1}{2}}K_{t})^{T}D_{t} + D_{t}(A{+}G^{\frac{1}{2}}K_{t}) \leq -F.
\end{align}
\label{xblemma}
\end{lemma}

From this lemma, given any stabilizing controller $K_{t}$ one can always find an upper bound for $\bar{\lambda}(X)$ using the function $\beta(A,G,Q)$. This function will be used next to verify stability of $\hat{g}(s)$.

\subsubsection{Stability condition}

Before stating the main stability criterion for $\hat{g}(s)$, we first make the following assumption on the projection weight $w$ from Definition \ref{Pd}.
\begin{assumption}
The weight vector $w$ satisfies $w^{T}_{\mathcal{I}_{i}}v_{\mathcal{I}_{i}} \neq 0$, $i=1,...,r$ for any $Av = \lambda v$, $Re(\lambda) \geq 0$.
\label{assw}
\end{assumption}

Note that from Definition \ref{Pd}, if $w^{T}_{\mathcal{I}_{i}}v_{\mathcal{I}_{i}} = 0$, $i=1,...,r$ then $Pv = 0$. This would imply $(A-B\hat{K})v = Av = \lambda v$, $Re(\lambda)\geq 0$, which means that the unstable eigenvalues of $A$ are retained in the closed-loop. Assumption \ref{assw} is made to avoid this situation. With Assumption \ref{assw}, we next state a sufficient condition for the stability of $\hat{g}(s)$.
\begin{theorem}
\label{stab}
Given Assumption \ref{asss} and \ref{assw}, the TFM $\hat{g}(s)$ is asymptotically stable if 
\begin{align}
\underline{\sigma}(Q) - 2\bar{\sigma}(X) \bar{\sigma}(G)\bar{\sigma}(E) + \underline{\sigma}(X)^{2} \underline{\sigma}(G)  > 0 ,
\label{stabcond}
\end{align}
where $E=X-\hat{X}$ denotes the error between the full-order ARE solution from (\ref{are}) and the projected solution from (\ref{invert}).
\end{theorem}

Theorem \ref{stab} provides two options for achieving stability of $\hat{g}(s)$. In order to meet the condition (\ref{stabcond}), one can select $Q$ and $R$ such that $\underline{\sigma}(Q) \gg \bar{\sigma}(G)>0$. In practice, this choice of $(Q, R)$ will make the design more robust towards the uncertainties in $A$ and $B$ \cite{robustLQR}. As a drawback, a larger $Q$ will also result in a high feedback gain, making the system vulnerable to noise. An alternative to satisfy (\ref{stabcond}) is to minimize $\bar{\sigma}(E)$. As will be shown in the next subsection, $\bar{\sigma}(E)$ is proportional to the objective function of our proposed upper bound relaxation for (\ref{MM}). Therefore, solving the relaxation problem will also assist in enlarging the inequality gap in (\ref{stabcond}). The following lemma provides a sufficient condition for (\ref{stabcond}) by which stability of $\hat{g}(s)$ can be verified without knowing $X$.
\begin{lemma}
The stability condition (\ref{stabcond}) holds if 
\begin{align}
 \bar{\sigma}(\tilde{X}) < \frac{\underline{\sigma}(Q)}{2 \bar{\sigma}(G)\beta(A,G,Q)} - \beta(A,G,Q), \label{vcondieq}
\end{align}
where function $\beta(A,G,Q)$ is defined in Lemma \ref{xblemma}.
\label{vcond}
\end{lemma} 

\begin{remark}
For the case when $A$ is unstable, computing $\beta(A,G,Q)$ would require the knowledge of $K_{t}$ according to Lemma \ref{xblemma}. Similarly for $w$ to meet Assumption \ref{assw}, one also needs to know the eigenvectors of all unstable modes of $A$. Verifying stability of $\hat{g}(s)$ for an unstable $A$ is, therefore, admittedly more computationally expensive than for a stable $A$. This computational burden obviously does not exist when $A$ is stable, and may not also exist when $A$, despite having zero eigenvalues, has specific structural properties. One such example is when $A$ is a weighted Laplacian matrix (i.e., when (1) is a consensus network). We will illustrate this special case of consensus network in Appendix A.  
\end{remark}

\subsection{Relaxation I: Upper Bound Minimization}

We present our first stage of relaxation assuming that the sufficient condition for closed-loop stability from Theorem \ref{stab} holds. As mentioned before, finding a $P$ that exactly minimizes $\| g(s)-\hat{g}(s) \|_{\mathcal{H}_{2}}$ is an intractable problem. To relax this problem, we find an upper bound for $\| g(s)-\hat{g}(s) \|_{\mathcal{H}_{2}}$. We denote the error system by $g_{e}(s)$, which can be written as
\begin{align}
\nonumber
g_{e}(s) :&= g(s) - \hat{g}(s) = \left[
\begin{array}{c|c}
A_{e} & B_{e} \\ \hline
C_{e} & 0
\end{array}
\right] \\
&= \left[
\begin{array}{cc|c}
A-GX & & B_{d} \\ 
 & A-G\hat{X} & -B_{d} \\ \hline
I_{n} & I_{n} & 0
\end{array}
\right].
\label{gerr}
\end{align}
From a similarity transformation of $T=\begin{bmatrix} I_{n} & I_{n} \\ I_{n} & 0 \end{bmatrix}$ and $T^{-1}=\begin{bmatrix} 0 & I_{n} \\ I_{n} & -I_{n} \end{bmatrix}$, (\ref{gerr}) yields
\begin{align}
\nonumber
& g_{e}(s)  {=} \left[
\begin{array}{c|c}
TA_{e}T^{-1} & TB_{e} \\ \hline
C_{e}T^{-1} & 0
\end{array}
\right] {=} \left[
\begin{array}{cc|c}
A{-}G\hat{X} & G(\hat{X}{-}X) & 0 \\ 
 & A{-}GX & B_{d} \\ \hline
I_{n} & 0 & 0
\end{array}
\right] \\ 
& = -(sI_{n}-A+G\hat{X})^{-1}GE(sI_{n}-A+GX)^{-1}B_{d}, \label{ger}
\end{align}
where $E=X-\hat{X}$. By taking norms on both sides of (\ref{ger}), we get
\begin{align*}
\| g_{e}(s) \|_{\mathcal{H}_{2}} \leq \| (sI_{n}-A+G\hat{X})^{-1}G \|_{\mathcal{H}_{\infty}} \| Eg(s)\|_{\mathcal{H}_{2}},
\end{align*}
which from the bounded real lemma \cite{rc} and the definition of $\mathcal{H}_{2}$ norm reduces to
\begin{align}
\| g_{e}(s) \|_{\mathcal{H}_{2}} \leq \gamma \| E\Phi^{\frac{1}{2}}\|_{F},
\label{bieq}
\end{align}
where $\gamma$ is any positive real number such that a real-valued matrix $\Gamma = \Gamma^{T} \succeq 0$ exists and satisfies
\begin{align*}
\begin{bmatrix}
\Gamma (A-G\hat{X})+(A-G\hat{X})^{T}\Gamma & \Gamma G & I_{n} \\
G\Gamma & -\gamma I_{n} & 0 \\
I_{n} & 0 & -\gamma I_{n}
\end{bmatrix} \prec 0,
\end{align*}
and $\Phi :=\Phi^{\frac{1}{2}} \Phi^{\frac{T}{2}} = \int_{0}^{\infty} e^{(A-GX)\tau}B_{d}B_{d}^{T}e^{(A-GX)^{T}\tau} d\tau \succ 0$ is the solution of the Lyapunov equation
\begin{align}
(A-GX)\Phi + \Phi (A-GX)^{T} + B_{d}B_{d}^{T} = 0.
\label{lyap}
\end{align}
Inequality (\ref{bieq}) shows that $\| g_{e}(s)\|_{\mathcal{H}_{2}}$ is linearly bounded by $\| E\Phi^{\frac{1}{2}}\|_{F}$. Therefore, one way to solve the original model matching problem (\ref{MM}) will be to find a $P$ that minimizes $\| E\Phi^{\frac{1}{2}}\|_{F}$. By doing so, $\bar{\sigma}(E)$ can also be minimized to some extent since $\bar{\sigma}(E) \leq \|E\|_{F} \leq \bar{\sigma}(\Phi^{-\frac{1}{2}}) \| E\Phi^{\frac{1}{2}}\|_{F}$, which will help in meeting the stability condition (\ref{stabcond}). This type of bound minimization is common in model and controller reduction, and has been attempted (see \cite{are} and the references therein) under the assumption that $P$ is unstructured, or more specifically $P$ is an $r$-dimensional Krylov subspace from $\mathcal{K}(A,Q^{\frac{1}{2}},r)=\begin{bmatrix}
Q^{\frac{1}{2}} & AQ^{\frac{1}{2}} & \cdots & A^{r-1}Q^{\frac{1}{2}}
\end{bmatrix}$. By this assumption, $E$ can be found as an explicit function of $P$ associated with a Householder transformation. In our case, however, $P$ has a structure as in (\ref{eqp}), due to which this explicit functional relationship does not hold anymore. We, therefore, apply perturbation theory of ARE to further relax the bound in (\ref{bieq}), and derive a new upper bound on $\| E\Phi^{\frac{1}{2}}\|_{F}$ as an explicit function of $P$ in the following theorem.

\begin{theorem}
Denote $\xi = \| (I_{n}-P^{T}P)\Phi^{\frac{1}{2}} \|_{F}$. The norm of the weighted error $E\Phi^{\frac{1}{2}}$ satisfies the inequality
\begin{align}
& \| E\Phi^{\frac{1}{2}}\|_{F} \leq f(\xi) = \epsilon_{1} \bar{\sigma}(Q) \xi^{2} + 2 \epsilon_{1} \epsilon_{2} \xi  + \alpha \epsilon_{1}\epsilon_{3},
\label{rboundn}
\end{align}
where $\epsilon_{1} = \frac{\bar{\sigma}(\Phi^{-\frac{1}{2}}) }{\underline{\sigma}[\Phi^{-\frac{1}{2}}(A-GX)\Phi^{\frac{1}{2}}]}$, $\epsilon_{2} = \tilde{\beta} \bar{\sigma}(A) \bar{\sigma}(\Phi^{\frac{1}{2}}) + \bar{\sigma}(Q\Phi^{\frac{1}{2}})$, $\epsilon_{3} = (\tilde{\beta}^{2} + 1 )\bar{\sigma}(\Phi) $, and $\tilde{\beta} = \mathrm{sup}_{P}\ \beta(\tilde{A},\tilde{G},\tilde{Q})$ are positive scalars that are independent of $P$, and $\alpha$ is defined in (\ref{QRmod}).
\label{tmain}
\end{theorem}

The constant $\alpha \epsilon_{1}\epsilon_{3}$ in (\ref{rboundn}) represents the matching error introduced by the constant shift $\alpha$ from (\ref{QRmod}). This error can be disregarded if the reduced-order ARE (\ref{reare}) admits a solution. From (\ref{bieq}) and (\ref{rboundn}), it then follows that $\| g_{e}(s)\|_{\mathcal{H}_{2}} \leq \gamma f(\xi)$, due to which we approach the minimization of $\| g_{e}(s)\|_{\mathcal{H}_{2}}$ by minimizing $f(\xi)$ with respect to $P$. Since $f(\xi)$ is a monotonic function of $\xi$, the minimization of $f(\xi)$ is equivalent to minimizing the value of $\xi$ as
\begin{equation}
\begin{aligned}
& \underset{P}{\mathrm{minimize}}
& &  \xi = \| \Phi^{\frac{1}{2}}-P^{T}P\Phi^{\frac{1}{2}}\|_{F}.
\end{aligned}
\label{xi} \tag{RL-1}
\end{equation}
The optimization (\ref{xi}), therefore, serves as an upper bound relaxation for the original model matching problem (\ref{MM}). Note that, in general, it is impossible to exactly quantify the optimality gap between (\ref{MM}) and (\ref{xi}) since (\ref{MM}) is non-convex even without posing any combinatorial constraints on $P$. The optimality gap will be small if the minimum value of $\xi$ is close to zero, in which case the error $\| g_{e}(s)\|_{\mathcal{H}_{2}}$ in (\ref{MM}) will be nearly zero as well following (\ref{bieq}) and (\ref{rboundn}). 

The ideal case $\xi = 0$ will happen when $\Phi^{\frac{1}{2}} = P^{T}P\Phi^{\frac{1}{2}}$, meaning that $\Phi^{\frac{1}{2}}$ is invariant to the mapping $P^{T}P$. This holds for the trivial case where $r=n$ and $P=I_{n}$. To achieve a sufficiently small minimum for (\ref{xi}) for $r< n$, we will develop two designs in Sections \RNum{4} and \RNum{5}.

\subsection{Relaxation II: Low-Rank Approximation}

We next discuss the numerical complexity in constructing the optimization problem (\ref{xi}), and how this complexity can be simplified by making appropriate approximations on $\Phi$. In the most general case, $\Phi$ required for (\ref{xi}) needs to be computed through the following procedures. First, recall the definition of the Hamiltonian matrix 
\begin{align}
H :=\begin{bmatrix}
A & -G \\
-Q & -A^{T}
\end{bmatrix}.
\label{hami}
\end{align}
The eigenvalues of $H$ are symmetric about the imaginary axis. Suppose $H$ is diagonalizable and that the columns of the matrix $\begin{bmatrix} Y \\ Z \end{bmatrix}_{2n\times n}$ span the stable invariant subspace of $H$, i.e.
\begin{align}
H\begin{bmatrix}
Y \\ Z
\end{bmatrix} = \begin{bmatrix}
Y \\ Z
\end{bmatrix}\Lambda^{-},
\label{hamieig}
\end{align}
where $\Lambda^{-} = diag([\lambda_{1}^{-},...,\lambda_{n}^{-}])$ consists of all the eigenvalues of $H$ in the left-half plane, i.e. $0>\lambda_{1}^{-}>...>\lambda_{n}^{-}$. The stabilizing solution of ARE can thus be found by $X=ZY^{-1}$ \cite{rc}. The first $n$ rows of (\ref{hamieig}) are expanded as $A - GZY^{-1} = Y\Lambda^{-}Y^{-1}$, which means $\Lambda^{-}$ and $Y$ are the eigenvalues and right eigenspace of the closed-loop state matrix $A-GX$. Then from the Lyapunov equation (\ref{lyap}), we can write $\Phi$ directly in terms of $Y$ and $\Lambda^{-}$ as \cite{antoulas}
\begin{align}
\Phi = Y(Y^{-1}B_{d}B_{d}^{T}Y^{-T}\circ \mathcal{C})Y^{T}, \label{gramsolu}
\end{align}
where $\mathcal{C} \in \mathbb{R}^{n\times n}$ is a Cauchy matrix with
$$\mathcal{C}_{i,j} =\begin{bmatrix}
-\frac{1}{\lambda_{i}^{-}+\lambda_{j}^{-}}
\end{bmatrix},$$ 
and subsequently obtain $\Phi^{\frac{1}{2}}$ from the Cholesky decomposition. Therefore, to compute $\Phi$ and then $\Phi^{\frac{1}{2}}$, one will need to compute the full stable eigenspace $Y$ from $H$ following (\ref{hamieig}). This computation is as expensive as solving a full-order LQR with $\mathcal{O}(n^{3})$ complexity for both computation and memory \cite{golub}. This may defeat the purpose of our design since we want our controller to be numerically much simpler than the full-order LQR. To bypass this difficulty, we next show that $\Phi$ can be approximated by a matrix $\Phi_{\kappa}$ that follows from a $\kappa$-dimensional ($\kappa < n$, not necessarily equal to $r$) invariant subspace of $Y$. Ideally $\kappa$ should be at most $r$ to justify the computational benefit of our design while preserving an acceptable accuracy in the error norm $\xi$. This matrix $\Phi_{\kappa}$ is constructed as follows. 
\begin{definition}
Define $\Phi_{\kappa}\in \mathbb{R}^{n\times n}$ as 
\begin{align}
\Phi_{\kappa} := \Phi_{\kappa}^{\frac{1}{2}}\Phi_{\kappa}^{\frac{T}{2}} =Y_{1}(\Omega_{1}B_{d}B_{d}^{T}\Omega_{1}^{T}\circ \mathcal{C}_{1,1})Y_{1}^{T},
\end{align} 
where $Y_{1} {=} Y_{:,1:\kappa}$, $\Omega_{1} {=} Y^{-1}_{1:\kappa,:}$ and $\mathcal{C}_{1,1} {=} \mathcal{C}_{1:\kappa,1:\kappa}$ are respectively the $\kappa$-dimensional partitions of $Y$, $Y^{-1}$ and $\mathcal{C}$. 
\label{Phik}
\end{definition}

By definition of $\Phi_{\kappa}$, one only needs to compute the first $\kappa$ eigenvalues $\lambda^{-}_{1},...,\lambda^{-}_{\kappa}$ of $H$, and the $Y_{1}$ component of the first $\kappa$ eigenvectors. $\Omega_{1}$ can be approximated by the pseudo-inverse of $Y_{1}$. These $\kappa$ smallest eigenvalues and eigenvectors can be solved by Krylov subspace-based techniques such as Arnoldi algorithm in $\mathcal{O}(n\kappa^{2})$ time \cite{golub}. Therefore, instead of (\ref{xi}), we consider solving a computationally simpler approximation of (\ref{xi}) as
\begin{equation}
\begin{aligned}
& \underset{P}{\mathrm{minimize}}
& &  \xi_{\kappa} = \| \Phi_{\kappa}^{\frac{1}{2}}-P^{T}P\Phi_{\kappa}^{\frac{1}{2}}\|_{F},
\end{aligned}
\label{xik} \tag{RL-2}
\end{equation}
where $\Phi_{\kappa}^{\frac{1}{2}} \in \mathbb{R}^{n\times \kappa}$ from Definition \ref{Phik} can be computed as $\Phi_{\kappa}^{\frac{1}{2}} = Y_{1}(\Omega_{1}B_{d}B_{d}^{T}\Omega_{1}^{T}\circ \mathcal{C}_{1,1})^{\frac{1}{2}}$. The optimality gap between optimizations (\ref{xi}) and (\ref{xik}) can be quantified by the following lemma.
\begin{lemma}
Assume $Y^{-1}$ has a moderate condition number $\eta$, and each column of $B_{d}$ has a unitary norm. The minimum $\xi^{*}$ from (\ref{xi}) and $\xi_{\kappa}^{*}$ from (\ref{xik}) satisfy 
\begin{align}
 \xi^{*} -  \xi_{\kappa}^{*} \ \leq \ \sqrt{\eta^{2} n_{b} \sum_{i=\kappa+1}^{n} -\frac{1}{2\lambda^{-}_{i}}}.
\label{xierror}
\end{align}
\label{kerror}
\end{lemma}

From Lemma \ref{kerror}, the optimality gap between (\ref{xi}) and (\ref{xik}) will be negligible when the error $\sqrt{\eta^{2}n_{b}\sum_{i=\kappa+1}^{n} -\frac{1}{2\lambda^{-}_{i}}}$ is kept small. In practice, this situation happens when there exists $\kappa$ ($\kappa\ll n$) dominant eigenvalues in the Hamiltonian matrix, i.e., $H$ has the following spectral gap
\begin{align*}
0 < |\lambda_{1}^{-}| < ... < |\lambda_{\kappa}^{-}| \ll  |\lambda_{\kappa+1}^{-}| < ... < |\lambda_{n}^{-}| .
\end{align*}
This gap can exist if the open-loop network (\ref{full}) exhibits coherent behavior \cite{chow}. As a result of this spectral gap, the RHS of (\ref{xierror}) can become sufficiently small, in which case the optimal value of (\ref{xik}) will closely resemble that of (\ref{xi}) while the computation of the objective function requiring a much tractable complexity of $\mathcal{O}(n\kappa^{2})$, $\kappa < n$. The idea of utilizing a $\kappa$-dimensional subspace for computing $\Phi_{\kappa}$ is similar in spirit to finding an unstructured approximate ARE solution as proposed in \cite{are}. However, it should be noted that unlike \cite{are} where the selection of the $\kappa$ eigenvectors is undetermined, for our problem the error bound in Lemma \ref{kerror} clearly guides the choice of the $\kappa$ eigenvectors in terms of tightening the optimality gap between (\ref{xi}) and (\ref{xik}).

\section{Design \RNum{1}: Cluster Design}

In this section we present an algorithm to design $P$ by solving (\ref{xik}). Note that $P$ has two degrees of freedom - $\mathcal{I}$ and $w$. For the design in this section, we keep $w$ fixed, and minimize (\ref{xik}) over $\mathcal{I}$. Although inherently this is an NP-hard problem, fortunately the specific structure of the objective function $\xi_{\kappa}$, together with the structure imposed on $P$ in (\ref{eqp}), enables (\ref{xik}) to be solved by efficient numerical algorithms such as weighted k-means \cite{kmeans}. We show these results as follows.

To establish the equivalency of (\ref{xik}) to the weighted k-means optimization, it is useful to borrow a nominal projection matrix $\bar{P}$ as
\begin{align}
\bar{P}_{i,j} := \begin{cases} 
\frac{1}{\| w_{\mathcal{I}_{i}} \|_{2}} & \quad  j \in \mathcal{I}_{i} \\ 
0 & \quad \text{otherwise} 
\end{cases},\ i=1,...,r.
\label{Pn}
\end{align}
From Definition \ref{Pd}, $\bar{P}$ satisfies $P=\bar{P}W$, where $W=diag(w)$. With this notation, we have 
\begin{align}
\Phi^{\frac{1}{2}}_{\kappa} - P^{T}P\Phi^{\frac{1}{2}}_{\kappa} = W(\Psi- \bar{P}^{T}\bar{P}W^{2}\Psi ),
\label{kwo}
\end{align}
where $\Psi  = [\psi_{1},...,\psi_{n}]^{T}$ denotes the matrix $W^{-1}\Phi^{\frac{1}{2}}_{\kappa}$. Therefore, the entries of the matrix $\bar{P}^{T}\bar{P}W^{2}$ can be found by
\begin{align*}
[\bar{P}^{T}\bar{P}W^{2}]_{j,k} = \begin{cases} 
\frac{w^{2}_{k}}{\| w_{\mathcal{I}_{i}}\|_{2}^{2}} & \quad  j \in \mathcal{I}_{i} \ \& \ k \in \mathcal{I}_{i}\\ 
0 & \quad \text{otherwise} 
\end{cases}
\end{align*}
for $i=1,...,r$. Thus the $j^{th}$ row of the matrix $\bar{P}^{T}\bar{P}W^{2}\Psi$ can be written as 
\begin{align}
[\bar{P}^{T}\bar{P}W^{2}\Psi]_{j,:} = c_{i}^{T} =  \frac{\sum_{k \in \mathcal{I}_{i}} w^{2}_{k}\psi_{k}^{T}}{\sum_{k \in \mathcal{I}_{i}} w^{2}_{k}},
\label{centroid}
\end{align}
for $j\in \mathcal{I}_{i}$. It is clear from above that for all the index $j$ that are assigned to the same cluster $\mathcal{I}_{i}$, $c_{i}^{T}$ is a weighted average (or a weighted centroid) of the row vectors $\psi_{j}^{T}$. Moreover, the matrix $\bar{P}^{T}\bar{P}W^{2}\Psi$ will have identical rows for those whose indices are inside the same cluster. Therefore, (\ref{xik}) can be posed as an unsupervised clustering problem
\begin{equation}
\begin{aligned}
& \underset{\mathcal{I}_{1},...,\mathcal{I}_{r}}{\mathrm{minimize}}
& &  \xi^{2}_{\kappa} = \sum_{j=1}^{n} w^{2}_{j}\| \psi_{j} - c_{i}\|_{2}^{2}.
\end{aligned}
\label{wkm}
\end{equation}
The optimization problem in (\ref{wkm}) is in the same form as a weighted k-means optimization, which minimizes the Euclidean distance weighted by $w^{2}_{j}$ between each data point $\psi_{j}$ and its centroid $c_{i}$. Thus, data points which are close to each other in the weighted distance are assigned to the same cluster. A standard method for solving this problem is Lloyd's algorithm \cite{kmeans}, using which we present the weighted k-means clustering for (\ref{wkm}) in Algorithm \ref{algk}. With $\Phi^{\frac{1}{2}}_{\kappa}\in \mathbb{R}^{n\times \kappa}$ as the input, the running time of Lloyd's algorithm is $\mathcal{O}(n\kappa rk)$, where $k$ is the total number of iterations. As in any heuristic algorithm, Algorithm \ref{algk} does not guarantee convergence to a global minimum. Hence, if one is not satisfied by the k-means solution one can apply any state-of-art clustering algorithms for solving (\ref{wkm}). Using the resulting clustering set $\mathcal{I}$ and the fixed weight $w$, we can then construct the projection $P$ and execute the control inversion design to get $\hat{K}$. The clustering weight $w$ can be selected as any vector that satisfies Assumption \ref{assw}. When $A$ is Hurwitz, a simple choice of $w$ can be the vector of all ones. 

We conclude this section by summarizing the total numerical complexity for our design based on Algorithm \ref{algk}. The chain of approximations involved in this design till the output of Algorithm \ref{algk} follows the sequence of equations: (\ref{bieq}), (\ref{rboundn}), (\ref{xi}), (\ref{xik}) and (\ref{wkm}). The total complexity amounts to $\mathcal{O}(n\kappa^{2}) + \mathcal{O}(n\kappa rk) + \mathcal{O}(n^{2}r) + \mathcal{O}(r^{3})$, which includes construction of the objective function for (\ref{xik}), execution of Algorithm \ref{algk}, computation of reduced-order matrices triple $(\tilde{A},\tilde{Q},\tilde{G})$, and solving the reduced-order LQR (\ref{LQRreduced}), respectively. This complexity can be at most $\mathcal{O}(n^{2}r)$ if $\kappa \leq r$, which is more tractable compared to the $\mathcal{O}(n^{3})$ complexity of full-order LQR, especially when $r$ and $\kappa$ are sufficiently small.

\begin{algorithm}
    \SetKwInOut{Input}{Input}
    \SetKwInOut{Output}{Output}
    \Input{$\Phi^{\frac{1}{2}}_{\kappa}$, $w$ and $r$}
 {\bf{Initialization}:} Assign $r$ random rows from $\Psi  = W^{-1}\Phi^{\frac{1}{2}}_{\kappa} = [\psi_{1},...,\psi_{n}]^{T}$ as the initial centroids $c_{1}^{0},...,c_{r}^{0}$\;
 Find initial clustering sets $\mathcal{I}^{0} = \{ j \to \mathcal{I}_{i}^{0} \ | \ \underset{i=1,...,r} {\mathrm{argmin}}\ w^{2}_{j} \|\psi_{j}-c_{i}^{0}\|_{2}^{2},\ j=1,...,n \}$\;
 Update the centroids: $c_{i}^{0}= \frac{\sum_{j \in \mathcal{I}_{i}^{0}} w^{2}_{j}\psi_{j}}{\sum_{j \in \mathcal{I}_{i}^{0}} w^{2}_{j}}$, $i=1,...,r$ \;
 $k=1$\;
 \While{$\mathcal{I}^{k-1}\neq \mathcal{I}^{k}$ or within maximum iterations}{
  Update clustering sets $\mathcal{I}^{k} = \{ j \to \mathcal{I}_{i}^{k} \ | \ \underset{i=1,...,r} {\mathrm{argmin}}\ w^{2}_{j}\|\psi_{j}-c_{i}^{k-1}\|_{2}^{2},\ j=1,...,n \}$\;
  Update the centroids: $c_{i}^{k}= \frac{\sum_{j \in \mathcal{I}_{i}^{k}} w^{2}_{j}\psi_{j}}{\sum_{j \in \mathcal{I}_{i}^{k}} w^{2}_{j}}$, $i=1,...,r$ \;
  $k=k+1$ \;
 }
 \Output{$\mathcal{I} = \mathcal{I}^{k}$}
 \caption{$\mathcal{H}_{2}$ closed-loop clustering} \label{algk}
\end{algorithm}

\section{Design \RNum{2}: Weight Design}

\begin{figure*}
    \centering
    \begin{subfigure}[t]{0.3\textwidth}
    \centering
        \includegraphics[width=0.8\columnwidth]{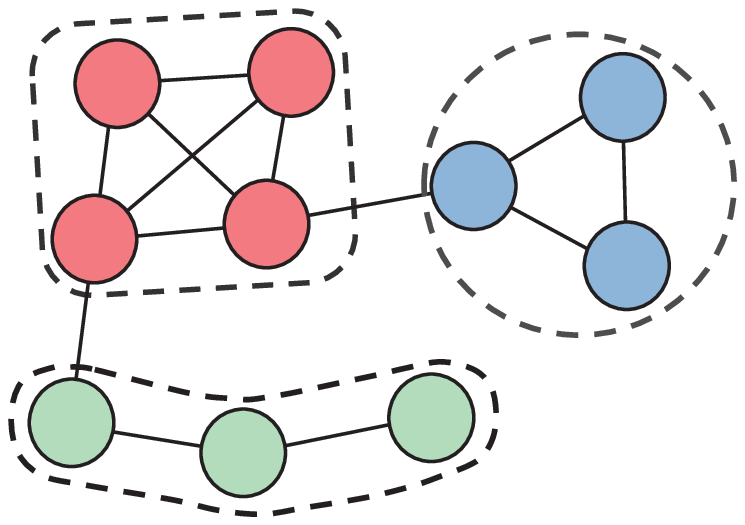}
        \caption{Spatial groups partitioned by the closeness in terms of geometric distances}
        \label{cw1}
    \end{subfigure}
    ~        
    \begin{subfigure}[t]{0.3\textwidth}
    \centering
        \includegraphics[width=0.8\columnwidth]{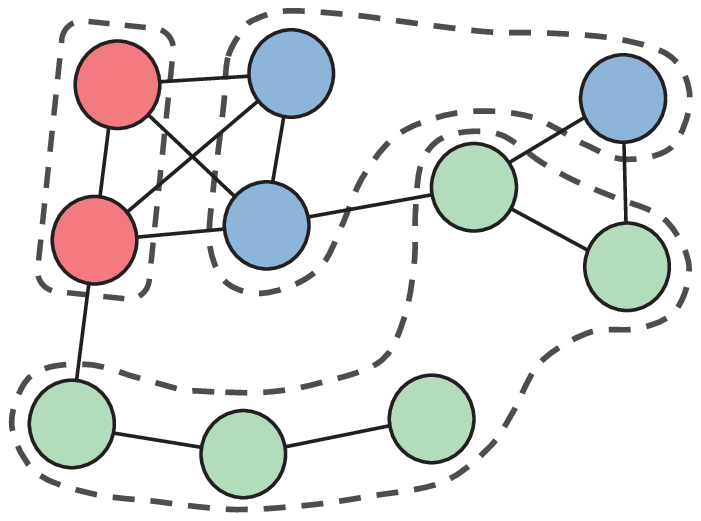}
        \caption{Node clusters resulting from cluster design as in Section \RNum{4}}
        \label{cw2}
    \end{subfigure}
    ~        
    \begin{subfigure}[t]{0.3\textwidth}
    \centering
        \includegraphics[width=0.8\columnwidth]{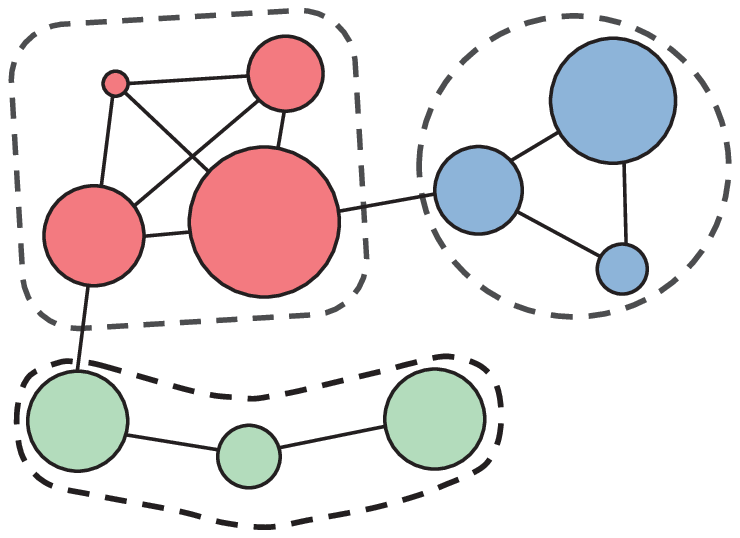}
        \caption{Imposing projection weights while fixing the clustering sets to the spatial clusters}
        \label{cw3}
    \end{subfigure}
    \caption{Illustration of cluster design and weight design.}\label{cw}
    \vspace{-1em}
\end{figure*}

We next state a variant of our proposed controller where we solve (\ref{xik}) by fixing the clustering set $\mathcal{I}$ and varying the projection weights in $w$ instead. This type of a controller may be needed when a subsystem prefers to have all of its states assigned to the same cluster. Multiple subsystems in the network may also exhibit spatial clustering based on their geographical proximities, in which case they may prefer to have a dedicated coordinator for themselves. This scenario commonly arises in power system networks. Utility companies always prefer to send the state information of their generators to only their own local control centers instead of sharing that information with any other company. Therefore, in this section we develop a new set of algorithms where we fix $\mathcal{I}$ to represent the identities of the desired clusters, and minimize $\xi_{\kappa}$ in (\ref{xik}) over $w$. A visual interpretation of this approach and its comparison to the design in Section \RNum{4} are shown in Fig. \ref{cw}. Note that the optimal values of $w$ so obtained denote the relative {\it importance} of the network nodes in the closed-loop system with the projected controller. In Fig. \ref{cw3} we show this relative importance by shrinking or expanding the size of the nodes. We describe the design for two cases depending on whether (\ref{full}) is stable or not.

\subsection{Case \RNum{1}: Stable Open-Loop}

We consider the same optimization objective as in (\ref{xik}) but now minimize it with respect to $w$ as
\begin{equation}
\begin{aligned}
& \underset{w}{\mathrm{minimize}}
& &  \xi_{\kappa} = \| \Phi_{\kappa}^{\frac{1}{2}}-P^{T}P\Phi_{\kappa}^{\frac{1}{2}}\|_{F}.
\end{aligned}
\label{xiw}
\end{equation}
To solve (\ref{xiw}), we consider a binary projection matrix $\hat{P}$ as
\begin{align*}
\hat{P}(i,j) := \begin{cases} 1 & \quad  j \in \mathcal{I}_{i} \\ 0 & \quad \text{otherwise}\\ \end{cases}.
\end{align*}
As can be verified from Definition \ref{Pd}, $\hat{P}$ here also satisfies $P=\hat{P}\hat{W}$ with $\hat{W}=diag(\hat{w})$, where this $\hat{w}$ is defined by
\begin{align}
\hat{w}_{j} = \frac{w_{j}}{\| w_{\mathcal{I}_{i}} \|_{2}},\ j \in \mathcal{I}_{i},
\label{wbar}
\end{align}
such that $\hat{w}_{\mathcal{I}_{i}}^{T}\hat{w}_{\mathcal{I}_{i}}=1$, for $i=1,...,r$. Using these notations, we can rewrite the objective function in (\ref{xiw}) as
\begin{align*}
\nonumber
\xi^{2}_{\kappa} & = tr(\Phi_{\kappa} - P^{T}P\Phi_{\kappa}) = tr(\Phi_{\kappa}) - tr(\hat{W}\hat{P}^{T}\hat{P}\hat{W}\Phi_{\kappa}) \\
& = tr(\Phi_{\kappa}) - \hat{w}^{T}(\hat{P}^{T}\hat{P}\circ \Phi_{\kappa})\hat{w}.
\end{align*}
Since $tr(\Phi_{\kappa})$ is a constant number, an equivalent form of (\ref{xiw}) follows as
\begin{equation}
\begin{aligned}
& \underset{\hat{w}}{\mathrm{maximize}}
& & \hat{w}^{T}(\hat{P}^{T}\hat{P}\circ \Phi_{\kappa})\hat{w} \\
& \mathrm{subject\ to}
& & \| \hat{w}_{\mathcal{I}_{i}}\|_{2}=1,\ i=1,...,r.
\end{aligned}
\label{op1}
\end{equation}
The Hadamard product $\hat{P}^{T}\hat{P}\circ \Phi_{\kappa}$ preserves the structure from $\hat{P}^{T}\hat{P}$, or equivalently from the clustering set $\mathcal{I}$ in the objective function. As a result, the optimization problem (\ref{op1}) boils down to $r$ decoupled optimizations
\begin{equation}
\label{op2} \tag{opt1}
\begin{aligned}
& \underset{\hat{w}_{\mathcal{I}_{i}}}{\mathrm{maximize}}
& & \hat{w}_{\mathcal{I}_{i}}^{T}\Phi_{\kappa[\mathcal{I}_{i},\mathcal{I}_{i}]}\hat{w}_{\mathcal{I}_{i}} \\
& \mathrm{subject\ to}
& & \| \hat{w}_{\mathcal{I}_{i}}\|_{2}=1,
\end{aligned}
\end{equation}
for $i=1,...,r$, where $\Phi_{\kappa[\mathcal{I}_{i},\mathcal{I}_{i}]}$ denotes the submatrix of $\Phi$ corresponding to the indices in $\mathcal{I}_{i}$. This decoupling can be illustrated by the same example we used before. 
\begin{example}
(Continued) Consider the sets $\mathcal{I}_{1}$, $\mathcal{I}_{2}$ and $\mathcal{I}_{3}$, and the matrix
\begin{align*}
\hat{P} = \begin{bmatrix}
1 & 1 & 0 & 0 & 0 & 0 & 0 & 0 & 0 & 0 \\
0 & 0 & 1 & 1 & 1 & 0 & 0 & 0 & 0 & 0\\
0 & 0 & 0 & 0 & 0 & 1 & 1 & 1 & 1 & 1
\end{bmatrix}.
\end{align*}
The objective function in (\ref{op1}) can be block diagonalized as %
\begin{small}
\begin{align*}
& \ \hat{w}^{T}(\hat{P}^{T}\hat{P} {\circ} \Phi_{\kappa})\hat{w} {=} \hat{w}^{T} \begin{bmatrix}
\Phi_{\kappa [1:2,1:2]} & 0 & 0 \\
0 & \Phi_{\kappa [3:5,3:5]} & 0 \\
0 & 0 & \Phi_{\kappa [6:10,6:10]} 
\end{bmatrix} \hat{w} {=} \\
& \begin{bmatrix}
\hat{w}^{T}_{1:2}\Phi_{\kappa [1:2,1:2]}\hat{w}_{1:2} & 0 & 0 \\
0 & \hat{w}^{T}_{3:5}\Phi_{\kappa [3:5,3:5]}\hat{w}_{3:5} & 0 \\
0 & 0 & \hat{w}^{T}_{6:10}\Phi_{\kappa [6:10,6:10]}\hat{w}_{6:10}
\end{bmatrix}.
\end{align*}
\end{small} %
\end{example}

Since Assumption \ref{assw} holds trivial for this case, we state the following theorem. 
\begin{theorem}
The global optimum for (\ref{op2}) is obtained at $\hat{w}_{\mathcal{I}_{i}} = \bar{v}(\Phi_{\kappa [\mathcal{I}_{i},\mathcal{I}_{i}}])$, for $i=1,..,r$.
\label{opsol}
\end{theorem}
\begin{proof}
Given $\hat{P}^{T}\hat{P} \succeq 0$ and $\Phi \succeq 0$, $\hat{P}^{T}\hat{P}\circ \Phi_{\kappa}$ is positive-semidefinite according to the Schur product theorem. Also, the objective function in (\ref{op2}) is a standard Rayleigh quotient for symmetric eigenvalue problem. Therefore, the maximum of objective function in (\ref{op2}) is obtained at the largest singular value of $\Phi_{\kappa [\mathcal{I}_{i},\mathcal{I}_{i}]}$. Since $\Phi_{\kappa [\mathcal{I}_{i},\mathcal{I}_{i}]} \succeq 0$, its largest singular value is the same as its largest eigenvalue $\bar{\lambda}(\Phi_{\kappa [\mathcal{I}_{i},\mathcal{I}_{i}]})$, and hence the optimum is obtained at its dominant eigenvector $\bar{v}(\Phi_{\kappa [\mathcal{I}_{i},\mathcal{I}_{i}]})$.
\end{proof}

Given that $\Phi_{\kappa [\mathcal{I}_{i},\mathcal{I}_{i}]}$ is symmetric, the eigenvector of its largest eigenvalue can be computed efficiently by Krylov subspace-based techniques, e.g. Lanczos algorithm \cite{golub} in this case. Solving (\ref{op2}) for $i=1,...,r$, therefore, requires a worst case complexity of $\mathcal{O}(n)$ in total.  Once $\hat{w}$ is solved from Theorem \ref{opsol}, one can then choose any $w$ that satisfies (\ref{wbar}), and this $w$ would serve as global optimum for (\ref{xiw}). 

\subsection{Case \RNum{2}: Unstable Open-Loop}

For an unstable open-loop system (\ref{full}), $w$ solved from Theorem \ref{opsol} is not guaranteed to conform to Assumption \ref{assw}, as a result of which $\hat{g}(s)$ can become unstable. To avoid such a hazardous situation, we add an extra penalty on (\ref{op2}) to restrict $w$ under Assumption \ref{assw}. This penalty term is formulated as follows. 

Recall the reduced-order ARE (\ref{reare}). By pre- and post-multiplying it with $P^{T}$ and $P$, and after a few calculations, we get $A^{T}\hat{X} + \hat{X}A + \hat{Q} - \hat{X}G\hat{X} = 0$, with
\begin{align}
\hat{Q} := P^{T}\tilde{Q}P - (U^{T}UA^{T}\hat{X} + \hat{X}AU^{T}U), \label{Qhat}
\end{align}
where $U$ is the complement of $P$. The ARE above implies that our proposed controller $\hat{K} = R^{-1}B^{T}\hat{X}$ is equivalent to an LQR problem $\mathrm{min}\ \int_{0}^{\infty} ( x^{T}\hat{Q}x+u^{T}Ru ) \mathrm{d}t$. Denote the eigenvalue decomposition of $A$ as
\begin{align}
A\begin{bmatrix}
V & \bar{V}
\end{bmatrix} = \begin{bmatrix}
V & \bar{V}
\end{bmatrix} \begin{bmatrix}
\Lambda & \\
 & \bar{\Lambda}
\end{bmatrix}, \label{vbar}
\end{align}
where $\bar{\Lambda} \succeq 0$ contains all the unstable eigenvalues of $A$. Applying the transformation $z = [V \ \ \bar{V}]^{-1}x$, one can write
\begin{align*}
x^{T}\hat{Q}x = z^{T}\begin{bmatrix}
V^{T}\hat{Q}V & V^{T}\hat{Q}\bar{V} \\
\bar{V}^{T}\hat{Q}V & \bar{V}^{T}\hat{Q}\bar{V}
\end{bmatrix}z.
\end{align*} 
The matrix $\bar{V}^{T}\hat{Q}\bar{V}$ corresponds to the LQR weight for the unstable dynamics of the state trajectories. Therefore, to assure that the unstable modes of $A$ are indeed penalized by $\hat{Q}$, we consider finding $w$ that maximizes $tr(\bar{V}^{T}\hat{Q}\bar{V})$. Given that the matrix $U^{T}UA^{T}\hat{X} + \hat{X}AU^{T}U$ in (\ref{Qhat}) is bounded but not sign-definite, we consider maximizing $tr(\bar{V}^{T}\hat{Q}\bar{V})$ by  simply maximizing its first part as
\begin{equation}
\begin{aligned}
& \underset{\hat{w}}{\mathrm{maximize}}
& & f_{\bar{V}} = tr(\bar{V}^{T}P^{T}\tilde{Q}P\bar{V}) \\
& \mathrm{subject\ to}
& & \| \hat{w}_{\mathcal{I}_{i}}\|_{2}=1,\ i=1,...,r.
\end{aligned}
\label{opt2} \tag{opt2}
\end{equation}
Combining (\ref{opt2}) with the main optimization (\ref{op2}), one can formulate the weight design for the unstable case as
\begin{equation}
\nonumber
\begin{aligned}
& \underset{\hat{w}}{\mathrm{maximize}}
& & [\hat{w}^{T}(\hat{P}^{T}\hat{P}\circ \Phi_{\kappa})\hat{w}]^{2} + \rho \cdot f_{\bar{V}} \\
& \mathrm{subject\ to}
& & \| \hat{w}_{\mathcal{I}_{i}}\|_{2}=1,\ i=1,...,r,
\end{aligned}
\label{opta} \tag{opt3}
\end{equation}
where $\rho >0$ is a penalty factor. In (\ref{opta}), the objective function from (\ref{op2}) is squared to match with the order of $f_{\bar{V}}$. This optimization problem is in the form of a fourth-order sum of squares (SOS) over $r>1$ sphere constraints, for which finding even a local optimal is very difficult. One way to bypass this can be to approximate matrix $P^{T}\tilde{Q}P$ in $f_{\bar{V}}$ as $P^{T}\tilde{Q}P \circ \hat{P}^{T}\hat{P}$, meaning only to retain the block-diagonal component of $Q$ only. In this way, (\ref{opta}) is reduced to $r$ SOS sub-problems with one sphere constraint for each as,
\begin{equation}
\begin{aligned}
& \underset{\hat{w}}{\mathrm{maximize}}
& &  (\hat{w}_{\mathcal{I}_{i}} \otimes \hat{w}_{\mathcal{I}_{i}})^{T} (\Phi_{\kappa [\mathcal{I}_{i},\mathcal{I}_{i}]} \otimes \Phi_{\kappa [\mathcal{I}_{i},\mathcal{I}_{i}]}  \\
& & &  + \rho Q_{\mathcal{I}_{i},\mathcal{I}_{i}}\otimes \bar{V}_{\mathcal{I}_{i},:}\bar{V}_{\mathcal{I}_{i},:}^{T})(\hat{w}_{\mathcal{I}_{i}} \otimes \hat{w}_{\mathcal{I}_{i}}) \\
& \mathrm{subject\ to}
& & \| \hat{w}_{\mathcal{I}_{i}}\|_{2}=1,
\end{aligned}
\label{optf} \tag{opt*} 
\end{equation}
for $i=1,...,r$.
While this approximation will follow naturally if $Q$ is block-diagonal, the upshot is that the closed-loop performance of the projected system may suffer if $Q$ has dominant off-block-diagonals. In practical networks, however, it is quite common to simply minimize the energy of a node itself, or the energy within a cluster, which implies that $Q$ is very commonly a diagonal or block diagonal matrix. In fact, $Q$ would indeed be preferred as block-diagonal for the scenario described in this section since network operators will always try to discourage closed-loop coupling of their own cluster nodes with other clusters. In those cases, (\ref{opta}) and (\ref{optf}) become equivalent problems, yielding the same solution. However, irrespective of whether $Q$ is block-diagonal or not, the following theorem shows that the solution of (\ref{optf}) will satisfy Assumption \ref{assw}.
\begin{theorem}
Given $\hat{w}_{\mathcal{I}_{i}}$, $i=1,...,r$ solved from (\ref{optf}) with $\rho > 0$, Assumption \ref{assw} holds for any $w$ satisfying (\ref{wbar}).
\end{theorem}
\begin{proof}
Note that $\bar{V}_{:,j}$, $j=1,...,n_{v}$ represents the eigenvector of the $j^{th}$ unstable eigenvalue. The second part of the objective function in (\ref{optf}) can be rewritten as $\rho \sum_{j=1}^{n_{v}} (\bar{V}^{T}_{\mathcal{I}_{i},j}\hat{w}_{\mathcal{I}_{i}})^{2}(\hat{w}_{\mathcal{I}_{i}}^{T}Q_{\mathcal{I}_{i},\mathcal{I}_{i}}\hat{w}_{\mathcal{I}_{i}})$. Once (\ref{optf}) is solved, each individual scalar $(\bar{V}^{T}_{\mathcal{I}_{i},j}\hat{w}_{\mathcal{I}_{i}})^{2}$ has to be maximized away from $0$, meaning $\bar{V}^{T}_{\mathcal{I}_{i},j}\hat{w}_{\mathcal{I}_{i}} \neq 0$, $i=1,...,r$, $j=1,...,n_{v}$ for $A\bar{V}_{:,j} = \lambda \bar{V}_{:,j}$, $Re(\lambda) \geq 0$. This satisfies Assumption \ref{assw}.
\end{proof}

The following lemma shows the performance trade-off due to the approximation in (\ref{optf}) when $Q$ is not block-diagonal.
\begin{lemma}
Let the maximum of (\ref{opta}) be $J_{1}$. The maximum $J_{2}$ of (\ref{optf}) satisfies
\begin{align}
J_{1} - J_{e} \leq J_{2} \leq J_{1} + J_{e},
\label{apperr}
\end{align}
where $J_{e} = \underset{j,l=1,...,n}{\underset{i=1,..,r}{\mathrm{max}}} \ \sum_{k=1,\ k\neq i}^{r} \rho \| Q_{\mathcal{I}_{i}(j),\mathcal{I}_{k}}\|_{1}\| \bar{V}_{l,:} \bar{V}^{T}_{\mathcal{I}_{k},:}\|_{1}$.
\label{circle}
\end{lemma}
\begin{proof}
The proof follows directly from the Gershgorin circle theorem \cite{golub}, and is shown in the Appendix.
\end{proof}

We next present the solution for (\ref{optf}). Since (\ref{optf}) is a set of $r$ decoupled problems, we illustrate the solution for just one cluster $\mathcal{I}=\{ \mathcal{I}_{1} \}$. This will also allow us to drop the subscripts in all the variables used in (\ref{optf}), making the notations easier to follow. We define a fourth-order tensor $\mathcal{F} \in \mathbb{R}^{n\times n\times n\times n}$ as
\begin{align}
\mathcal{F}_{i,j,k,l} = \Phi_{\kappa [i,j]}\Phi_{\kappa [k,l]} + \rho Q_{i,j}S_{k,l},\ i,j,k,l{=}1,...,n. \label{tensord1}
\end{align}
where $S$ denotes the product matrix $S = \bar{V}\bar{V}^{T}$. After a few manipulations, it can be shown that (\ref{optf}) is equivalent to the following problem
\begin{equation}
\begin{aligned}
& \underset{\hat{w}}{\mathrm{maximize}}
& & \mathcal{F} \odot (\hat{w}\otimes \hat{w} \otimes \hat{w} \otimes \hat{w}) \\
& \mathrm{subject\ to}
& & \| \hat{w}\|_{2}=1,
\end{aligned}
\label{opt3} \tag{opt**}
\end{equation}
where $\odot$ denotes the element-wise product. It has been studied in \cite{tensor1} that such a polynomial optimization is equivalent to finding the largest $Z$-eigenvalue of $\mathcal{F}$, if $\mathcal{F}$ is super-symmetric. From the definition in \cite{tensor1}, a super-symmetric tensor is one whose entries are invariant to any permutation to the index, i.e. $\mathcal{F}_{i,j,k,l}=...=\mathcal{F}_{l,k,j,i}$, which fails for (\ref{tensord1}) as $\mathcal{F}_{i,j,k,l}\neq \mathcal{F}_{i,k,j,l}$. However, note that although $\mathcal{F}$ is not super symmetric, $\mathcal{F} \odot (\hat{w}\otimes \hat{w} \otimes \hat{w} \otimes \hat{w})$ is a one-dimensional polynomial which is invariant to any index permutations.\footnote{This is analogous to an unsymmetric matrix whose quadratic form is invariant to the transpose operation, i.e. $z^{T}\frac{F^{T}+F}{2}z = z^{T}Fz = z^{T}F^{T}z \in \mathbb{R}$.} Following this logic, we rewrite the objective function in (\ref{opt3}) as follows.
\begin{proposition}
Given the fourth-order tensor $\mathcal{F}$ specified by (\ref{tensord1}), the polynomial $\mathcal{F} \odot (\hat{w}\otimes \hat{w} \otimes \hat{w} \otimes \hat{w})$ is identical to
\begin{align*}
\mathcal{F}^{s} \odot (\hat{w}\otimes \hat{w} \otimes \hat{w} \otimes \hat{w}) = (\hat{w}\otimes \hat{w})^{T} F^{s} (\hat{w} \otimes \hat{w}),
\end{align*}
where $\mathcal{F}^{s}$ is a super-symmetric tensor specified by
\begin{align*}
& \mathcal{F}^{s}_{i,j,k,l} = \frac{1}{3}(\Phi_{\kappa [i,j]}\Phi_{\kappa [k,l]} + \Phi_{\kappa [i,k]}\Phi_{\kappa [j,l]} + \Phi_{\kappa [i,l]}\Phi_{\kappa [j,k]}) + \\
& \frac{1}{6} \rho (Q_{i,j}S_{k,l} {+} Q_{i,k}S_{j,l} {+} Q_{i,l}S_{j,k} {+} Q_{j,k}S_{i,l} {+} Q_{j,l}S_{i,k} {+} Q_{k,l}S_{i,j})
\end{align*}
for $i,j,k,l=1,...,n$, and $F^{s} \in \mathbb{R}^{n^{2}\times n^{2}}$, the matrix unfolding of $\mathcal{F}^{s}$, can be obtained from
\begin{align*}
F^{s}_{n(i-1)+k,n(j-1)+l} = \mathcal{F}^{s}_{i,j,k,l},\quad i,j,k,l=1,...,n.
\end{align*}
\label{equivt}
\end{proposition}
The proof is omitted as the equations above can be easily verified by matching the coefficients of the polynomials on both sides. 

In summary, the optimization problem (\ref{opt3}) can be approached by substituting $\mathcal{F}$ with a super-symmetric tensor $\mathcal{F}^{s}$. One can, thereafter, solve (\ref{opt3}) using techniques developed for $Z$-eigenvalue problems. We solve (\ref{opt3}) using the tensor power iteration method \cite{tensor3} in Algorithm \ref{algp}. The convergence properties of this algorithm can be found in \cite{tensor3}. Due to the super symmetry of $\mathcal{F}^{s}$, the worst case (only one cluster) complexity for each iteration of Algorithm \ref{algp} is $\mathcal{O}(n^{4})$. Although this computation cost is expensive, the algorithm can be easily parallelized, and is easier to implement than $\mathcal{O}(n^{3})$ full-order LQR as the memory required is only $\mathcal{O}(n^{2})$. Moreover, the value of $n$ for Algorithm \ref{algp} scales down as the number of clusters increases.  
\begin{algorithm}[h]
    \SetKwInOut{Input}{Input}
    \SetKwInOut{Output}{Output}
    \Input{$\Phi_{\kappa}$, $Q$, $\bar{V}$, $\mathcal{I}$, $\rho$ and $\delta$}
    Partition $\Phi_{\kappa [\mathcal{I}_{i},\mathcal{I}_{i}]}$, $Q_{\mathcal{I}_{i},\mathcal{I}_{i}}$ and $\bar{V}_{\mathcal{I}_{i}}$ based on $\mathcal{I}$;\\
    \For{$i=1,...,r$}{
    Find $F^{s}$ corresponding to $\mathcal{I}_{i}$ by Proposition \ref{equivt};\\
    {\bf Initialization:} Compute the dominant eigenvector of $F^{s}$ as $\bar{v}(F^{s})$, then choose the initial vector $v^{0}$ as $v^{0}=\bar{v}(unvec(\bar{v}(F^{s})))$;\\
        $k=1$;\\
   \While{$\frac{(v^{k}\otimes v^{k})^{T}F^{s}(v^{k}\otimes v^{k})}{(v^{k-1}\otimes v^{k-1})^{T}F^{s}(v^{k-1}\otimes v^{k-1})}-1 > \delta$ or within maximum iterations}{
  $v^{k}=unvec(F^{s}(v^{k-1}\otimes v^{k-1}))v^{k-1}$;\\
  $v^{k}=\frac{v^{k}}{\| v^{k}\|_{2}}$;\\
  $k=k+1$;\\
 }
 $\hat{w}_{\mathcal{I}_{i}} = v^{k}$;\\
    }
    Construct $w$ and then $P$ by (\ref{wbar});\\
    \Output{$P$}
    \caption{Power iteration for projection weight design} \label{algp}
\end{algorithm}

\subsection{Optimizing (\ref{xik}) with respect to both $\mathcal{I}$ and $w$}
The designs proposed in Section \RNum{4} and this section can be combined to optimize (\ref{xik}) as a function of both $\mathcal{I}$ and $w$ iteratively. In this case, one would start with an arbitrarily chosen $w$, and minimize $\xi_{\kappa}$ with respect to $\mathcal{I}$ using Algorithm \ref{algk}. Say, the optimal cluster set is given as $\mathcal{I}^{\ast}$. Thereafter, one would fix $\mathcal{I}$ to $\mathcal{I}^{\ast}$, and minimize $\xi$ with respect to $w$ using Theorem \ref{opsol} or Algorithm \ref{algp} depending on whether (\ref{full}) is stable, and so on. The resulting algorithm is shown in Algorithm \ref{alg}. 

\begin{algorithm}[h]
    \SetKwInOut{Input}{Input}
    \SetKwInOut{Output}{Output}
    \Input{$A$, $B$, $B_{d}$, $Q$, $R$ and $r$}
    Compute $\Phi_{\kappa}^{\frac{1}{2}}$ by Definition \ref{Phik};\\
    Choose $w^{0}=\mathbf{1}_{n}$, and compute $W^{0}=diag(w^{0})$ and the k-means input $\Psi^{0}=(W^{0})^{-1}\Phi_{\kappa}^{\frac{1}{2}}$;\\
    $k=1$;\\
         \While{$\mathcal{I}^{k-1}\neq \mathcal{I}^{k}$ or within maximum iterations}{
  Solve $\mathcal{I}^{k}$ from Algorithm \ref{algk} by $(\Psi^{k-1},w^{k-1},r)$ \;
  Update $w^{k}$ from Theorem \ref{opsol} or Algorithm \ref{algp} by $\mathcal{I}^{k}$ \;
  $k=k+1$ \;
 }
    \Output{$P$}
    \caption{Iterative algorithm for finding $P$} \label{alg}
\end{algorithm}

\section{Numerical Examples}

\begin{figure*}
    \centering
    \begin{subfigure}[c]{1\columnwidth}
    \centering
        \includegraphics[width=0.8\columnwidth]{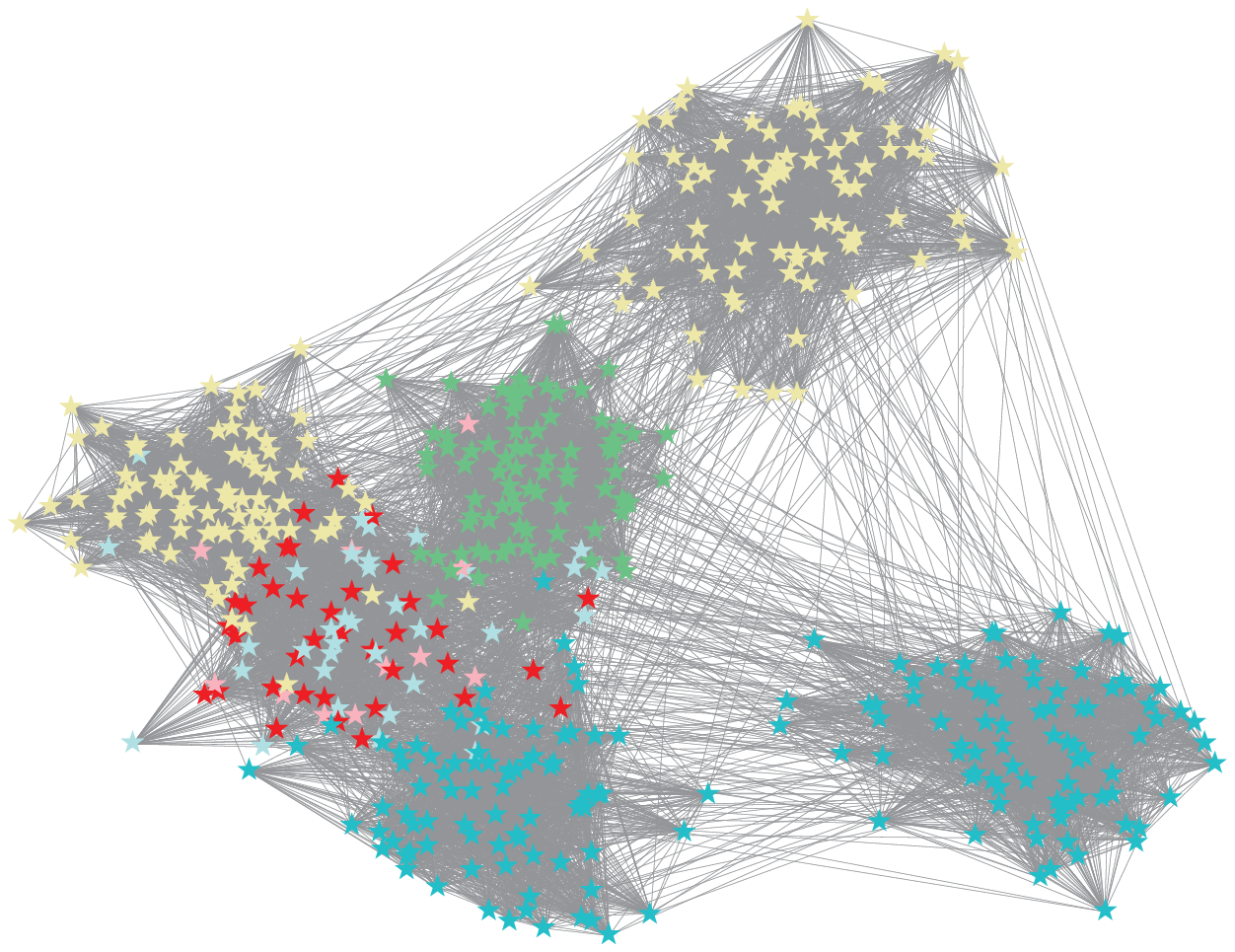}
        \caption{Clusters by coherency clustering}
        \label{ex1p3}
    \end{subfigure}
    ~     
    \begin{subfigure}[c]{1\columnwidth}
    \centering
        \includegraphics[width=0.8\columnwidth]{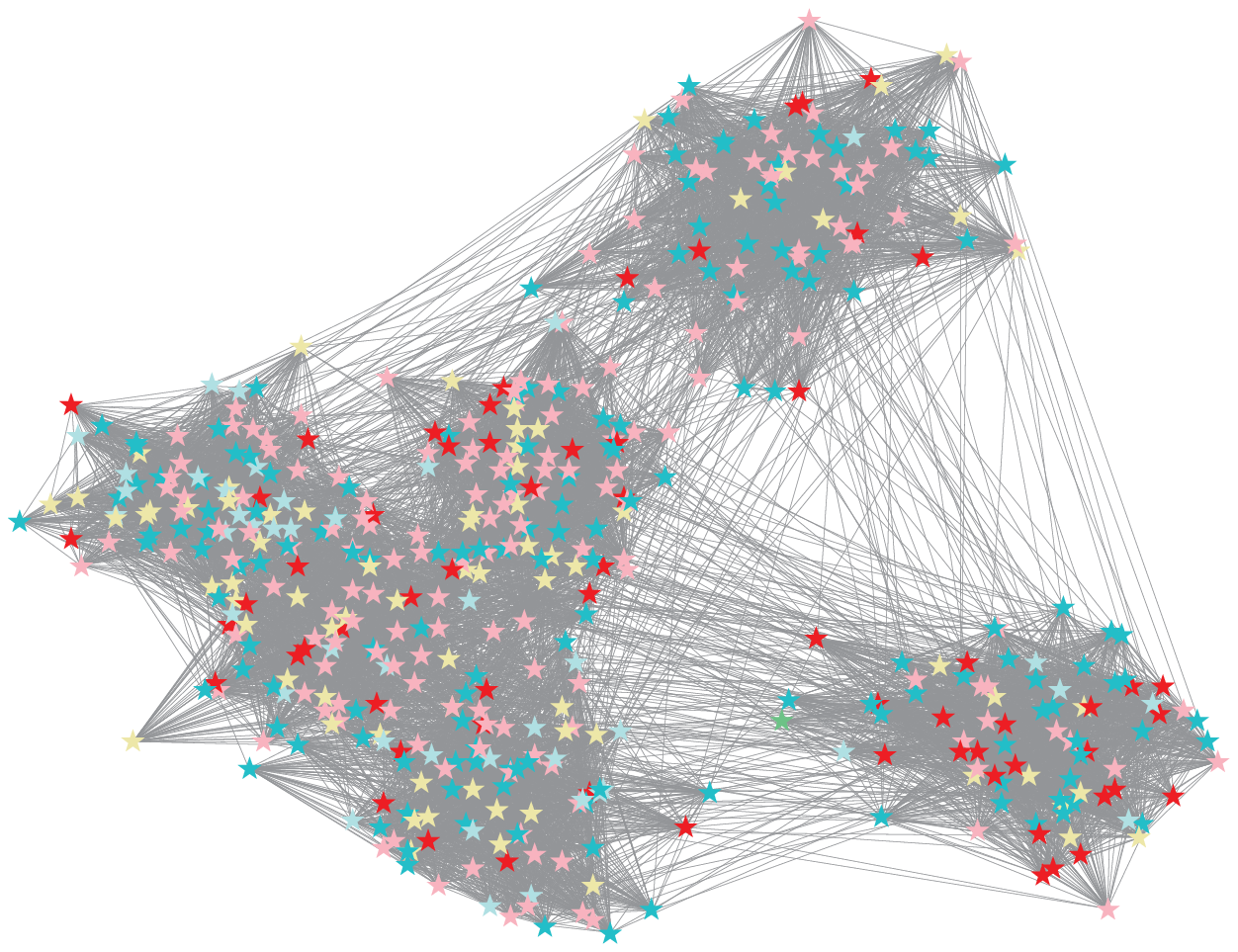}
        \caption{Clusters by $\mathcal{H}_{2}$ open-loop clustering}
        \label{ex1p2}
    \end{subfigure}
    ~
    
    \begin{subfigure}[r]{1\columnwidth}
    \centering
        \includegraphics[width=0.8\columnwidth]{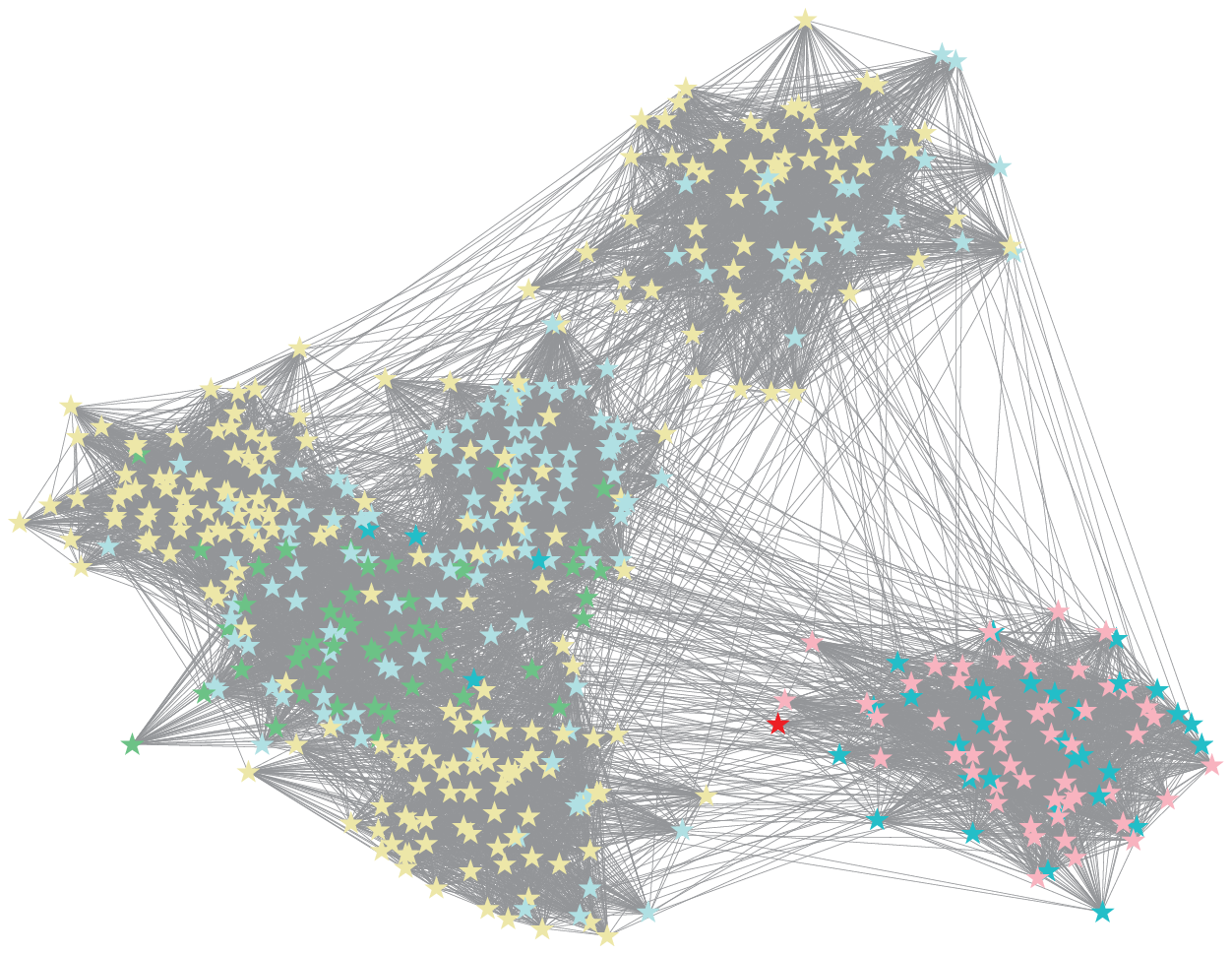}
        \caption{Clusters by $\mathcal{H}_{2}$ closed-loop clustering with $Q_{1}$}
        \label{ex1p1}
    \end{subfigure}
    ~  
    \begin{subfigure}[l]{1\columnwidth}
    \centering
        \includegraphics[width=0.8\columnwidth]{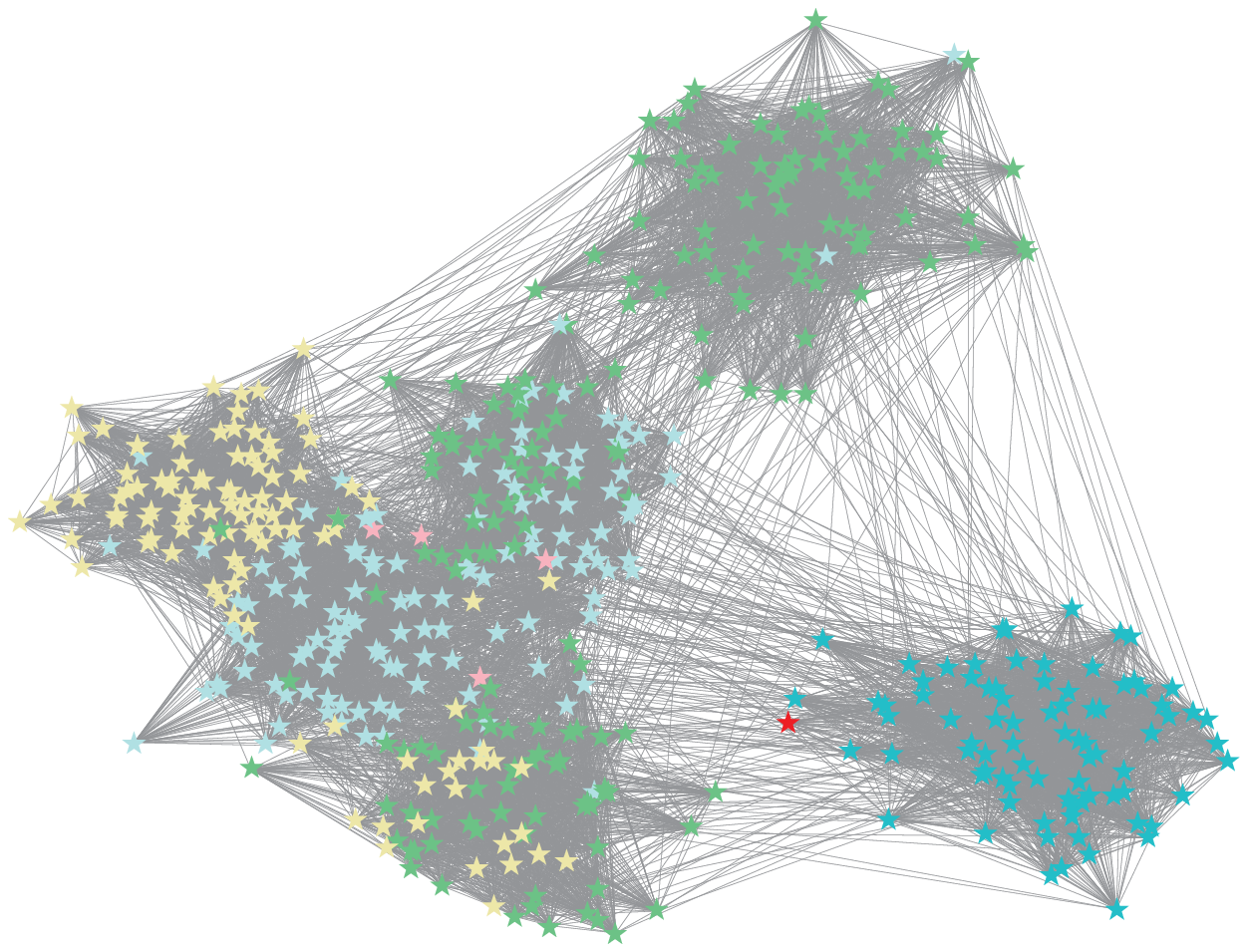}
        \caption{Clusters by $\mathcal{H}_{2}$ closed-loop clustering with $Q_{2}$}
        \label{ex2p1}
    \end{subfigure}    
    \caption{Clustering of the $500$-node network with $r = 6$ clusters. Nodes assigned to the same cluster are marked by the same color in one figure. Note that only the node identities are comparable between figures, but not the cluster identities.}
    \label{clusterplot}
\end{figure*}

To verify our proposed algorithms, we use a consensus network model defined over a $500$-node graph $\mathcal{G}_{500}$. The graph is randomly generated with $0.5$ overall probability for edge attachment, and $6$ spatial clusters with a proportion of $100$ for the number of edges within clusters versus the number of edges across clusters. We also apply a random weight $1\leq M_{i,i} \leq 2$ on each node. The resulting state matrix $A$ follows the expression (\ref{consensus}) presented in Appendix A. The disturbance is assumed to enter from the $364^{th}$ node, i.e. $B_{d}$ equals to the $364^{th}$ column of $I_{500}$. We assume $B=R=I_{500}$ and two choices of $Q$ as: a scaled identity matrix $Q_{1}=1000\times I_{500}$, and $Q_{2}=[L(\mathcal{G}_{500})]^{2}$, which is the square of the unweighted Laplacian matrix of $\mathcal{G}_{500}$. Both $Q_{1}$ and $Q_{2}$ satisfy Assumption \ref{asss}. For this simulation example, solvability of (\ref{reare}) and stability of $\hat{g}(s)$ are guaranteed by Theorem \ref{solvability} and Theorem \ref{gstab2} in Appendix A, respectively.

\subsection{Cluster Design}
\begin{figure*}
    \centering
    \begin{subfigure}[c]{1\columnwidth}
    \centering
        \includegraphics[width=0.85\columnwidth]{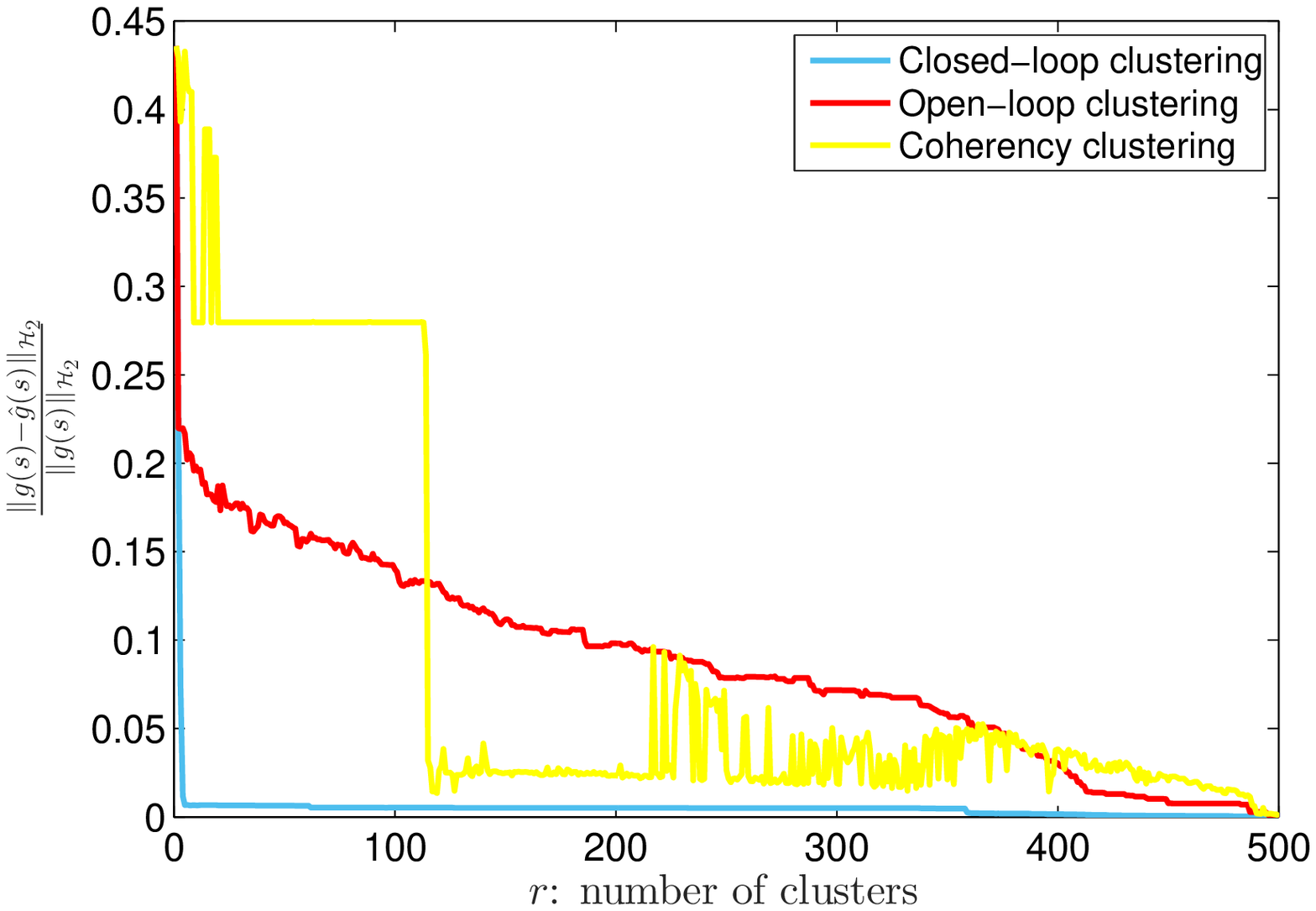}
        \caption{Design with $Q_{1}$}
        \label{eex1p3}
    \end{subfigure}
    ~         
    \begin{subfigure}[c]{1\columnwidth}
    \centering
        \includegraphics[width=0.85\columnwidth]{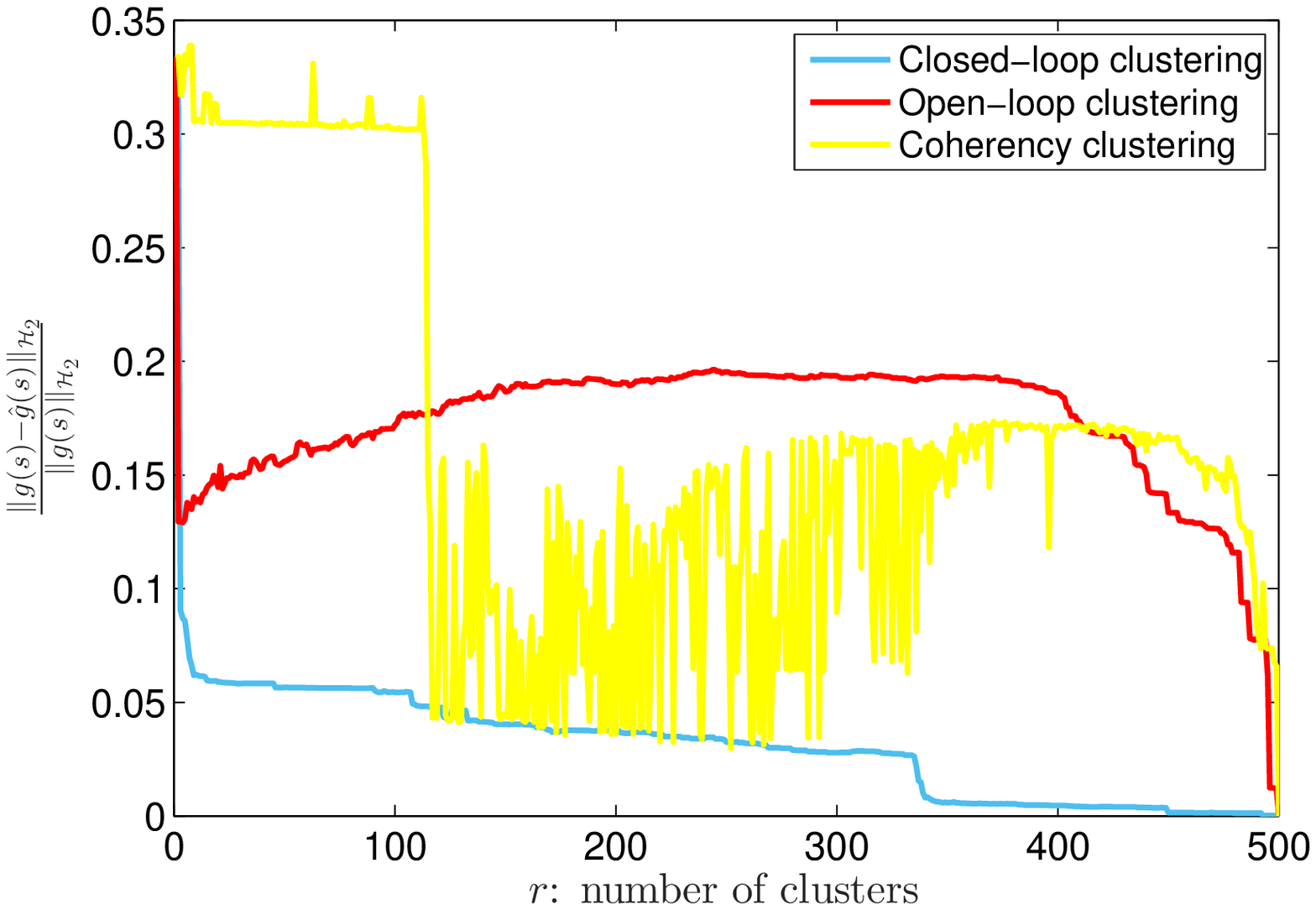}
        \caption{Design with $Q_{2}$}
        \label{eex1p2}
    \end{subfigure}
    \caption{Performance loss $\frac{\| g(s)-\hat{g}(s) \|_{\mathcal{H}_{2}}}{\| g(s)\|_{\mathcal{H}_{2}}}$ with respect to the number of clusters}
    \label{clpplot}
\end{figure*}

We start by finding the closed-loop clustering set $\mathcal{I}$ with respect to a fixed weight $w = \bar{v} = \bar{v}(A)$ with $r = 6$. For comparison, we also apply two other popular clustering algorithms, namely, $\mathcal{H}_{2}$ open-loop clustering proposed in \cite{h2}, and coherency-based clustering proposed in \cite{chow}. Both of these clustering algorithms can be transformed into Algorithm \ref{algk}, with their equivalent inputs as summarized in Table \ref{table1}.
\begin{table}
\caption{Input specifications for Algorithm \ref{algk} }
\label{table1}
\begin{threeparttable}
\centering
\begin{tabular}{c|ccc}
\hline \\
Algorithm & Data & Weight & Number of Clusters  \\ \hline
$\mathcal{H}_{2}$ closed-loop clustering & $\Phi$ & $w$ & $r$ \\ 
$\mathcal{H}_{2}$ open-loop clustering \cite{h2} & $\Phi_{o}$\tnote{a} & $w$ & $r$ \\ 
Coherency clustering \cite{chow} &  $\Phi_{c}$\tnote{b} & $\mathbf{1}_{n}$ & $r$\\ \hline
\end{tabular}
\begin{tablenotes}
\item[a] Let $v_{c}^{T}$ be the complement of $\bar{v}^{T}$, and then $\Phi_{o} :=\Phi_{o}^{\frac{1}{2}}\Phi_{o}^{\frac{T}{2}}=v_{c} [\int_{0}^{\infty} e^{(v_{c}^{T}Av_{c})\tau}v_{c}^{T}B_{d}B_{d}^{T}v_{c}e^{(v_{c}^{T}A^{T}v_{c})\tau} d\tau ]v_{c}^{T}$.
\item[b] Let the eigenvalues of $-L(\mathcal{G})$ be $0=\lambda_{1}>\lambda_{2}\geq ... \geq \lambda_{n}$, $\Psi_{c} = [v_{1},...,v_{r}]$ where $v_{i}$ is the right eigenvector of $\lambda_{i}$.
\end{tablenotes}
\end{threeparttable}
\end{table}  
Note that these two algorithms capture only the open-loop characteristics of the network, and hence do not depend on the choice of $Q$ and $R$. Fig. \ref{ex1p3} shows that the clusters identified by coherency based clustering closely resemble the spatial clusters of the open-loop network except for a few discrepancies. For example, two distant groups of nodes are assigned to the same cluster shown in yellow. These discrepancies arise from the fact that the spatial clusters are only based on the edge-weights (that model geographical distance between two nodes), while coherent clusters are decided by both edge-weights and node-weights. The $\mathcal{H}_{2}$ open-loop clusters are shown in Fig. \ref{ex1p2}. As such, they do not follow any definite pattern with respect to the spatial clusters as they are based on node aggregation following from the $\mathcal{H}_{2}$-norm distance of their output responses. Figures \ref{ex1p1} and \ref{ex2p1}, on the other hand, show the clusters identified by our $\mathcal{H}_{2}$ closed-loop algorithm (Algorithm \ref{algk}) for $Q=Q_{1}$ and $Q=Q_{2}$, respectively. Both of these clusters are different from each other for obvious reasons. They are also different from the spatial clusters, the coherent clusters as well as the $\mathcal{H}_{2}$ open-loop clusters as Algorithm \ref{algk} is related to the closed-loop controllability subspace.

We also illustrate the effectiveness of $\mathcal{H}_{2}$ closed-loop clustering with respect to the number of clusters $r$. As evident from the design, the error between the transfer matrices in (\ref{fullclptf}) and (\ref{clptff}) will be minimal when $r=n$, and will degrade with decreasing $r$ while improving tractability of the design. We vary $r$ from $1$ to $500$, and calculate the ratio $\frac{\| g(s)-\hat{g}(s) \|_{\mathcal{H}_{2}}}{\| g(s) \|_{\mathcal{H}_{2}}}$ resulting from the three clustering algorithms. The results are shown in Fig. \ref{clpplot}. For both $Q_{1}$ and $Q_{2}$ the closed-loop clustering outperforms the other two methods in approaching the $\mathcal{H}_{2}$ performance of $g(s)$. Therefore, even for very small values of $r$, the projected controller achieves significantly close $\mathcal{H}_{2}$ performance as the full-order LQR controller. In terms of implementation, the projected controller needs far less number of communication links than a full-order standard LQR as well as a full-order $\mathcal{L}_1$ sparsity-promoting LQR \cite{sparse}. For example, for this system a standard LQR would require ${{500}\choose{2} }=124750$ links. Meanwhile as shown in Fig. \ref{sparfig}, a sparsity promoting LQR requires from $3104$ to as many as $21325$ links to retain a performance loss under $5\%$. By choosing $r\leq 9$, the similar performance loss can be maintained by our design using at most $536$ links.

It is also noted that the $\mathcal{H}_{2}$ closed-loop clusters do not need to strictly follow the spatial geometric clustering of the network. For example, in both Figures (\ref{ex1p1}) and (\ref{ex2p1}), a cluster can be one single node as shown by the red, or can be scattered over the network such as yellow. In practice, this means that to implement the proposed control law, nodes from different geographical locations may need to be part of the same cluster for the closed-loop model, i.e., nodes that belong to two different spatial clusters in open-loop may need to collaborate and send their states to a common coordinator. The assignment, therefore, encourages system-wide participation from nodes at various corners of $\mathcal{G}_{500}$ for implementing the controller.
\begin{figure}
\includegraphics[width=1\columnwidth]{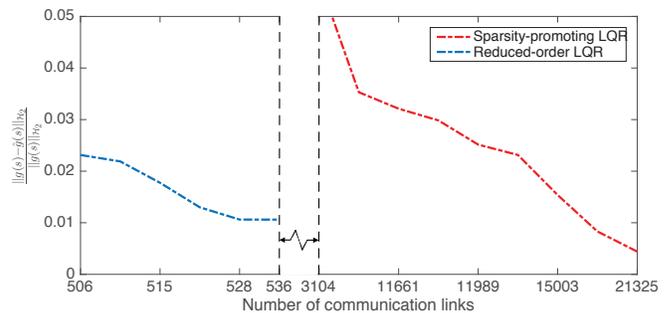}
        \caption{Projected LQR controller from clustering vs. LQR controller from sparsity-promoting algorithm with $Q_{1}$.}
        \label{sparfig}
\end{figure}

\subsection{Weight Design}
We next apply Algorithm \ref{algp} on $\mathcal{G}_{500}$ to find the optimal projection weight $w$ while fixing the clusters to those obtained from coherency. These clusters as previously shown in Fig. \ref{ex1p3} closely resemble the spatial clusters, and their clustering sets are given by $\mathcal{I}_{1}=\{ 1,...,167\}$, $\mathcal{I}_{2}=\{ 168,...,178\}$, $\mathcal{I}_{3}=\{ 179,...,344\}$, $\mathcal{I}_{4}=\{ 345,...,379\}$, $\mathcal{I}_{5}=\{ 380,...,463\}$ and $\mathcal{I}_{6}=\{ 464,...,500\}$. After running the algorithm with both $Q_{1}$ and $Q_{2}$, we plot the two weight vectors compared with $\bar{v} = \bar{v}(A)$ in Fig. \ref{wplot}. It can be seen that the weight vectors from $Q_{1}$ and $Q_{2}$ are very different than $\bar{v}$ or between themselves. On the other hand, both of these weights at the $364^{th}$ node, i.e., the node where the disturbance enters, show a sudden jump in magnitude from the rest of the nodes. To verify the closed-loop performance, we construct the $P$ matrices using these two vectors, and summarize the error ratios with some design parameters in Table \ref{wresult}. As expected, by applying the weight design, the closed-loop errors as shown in Table \ref{wresult} are significantly reduced from $w=\bar{v}$. Despite the fact that these two errors are much larger than what we get from closed-loop clustering (which yields an error of $0.68\%$), the weight design still grants us with significant improvement over the hard constraint on $\mathcal{I}$.

Finally, we compare the closed-loop performance of the iterative Algorithm \ref{alg} (where both $w$ and $\mathcal{I}$ are free) with Algorithm \ref{algk} in Fig. \ref{iterplot}. The comparison is shown for $Q_{1}$ and only $r\leq 6$ as the error ratio already becomes under $1\%$ after $r=6$. For this example, it is worth mentioning that Algorithm \ref{alg} turns out to be surprisingly efficient as it converges right after the first iteration. In this sense, the iterative process reduces to a single weight design after the clustering. Fig. \ref{iterplot} verifies that Algorithm \ref{alg} achieves better matching between $g(s)$ and $\hat{g}(s)$ than Algorithm \ref{algk}.

\begin{table}
\caption{Results of weight design}
\label{wresult}
\begin{threeparttable}
\centering
\begin{tabular}{ |c|c|c|c| }
\hline
\multirow{2}{*}{Case} & \multirow{2}{*}{Penalty factor $\rho$} & \multicolumn{2}{|c|}{Relative error $\| g(s)-\hat{g}(s) \|_{\mathcal{H}_{2}} / \| g(s)\|_{\mathcal{H}_{2}}$} \\
\cline{3-4}
&  & with $w=\bar{v}$ & with $w$ from Algorithm \ref{algp} \\ \hline
$Q_{1}$ & $0.011 \frac{\| \Phi \|_{2}}{\bar{v}^{T}Q\bar{v}}$ & $29.14\%$ & $7.35\%$ \\ \hline
$Q_{2}$ & $0.007 \frac{\| \Phi \|_{2}}{\bar{v}^{T}Q\bar{v}}$ & $35.69\%$ & $21.79\%$ \\ \hline
\end{tabular}
\begin{tablenotes}
\item[1] The convergence threshold for power iteration is chosen as $\delta = 0.05$.
\item[2] $\frac{\| \Phi \|_{2}}{\bar{v}^{T}Q\bar{v}}$ is included in the penalty factor to normalize the two objective functions in (\ref{optf}) to the same scale.
\end{tablenotes}
\end{threeparttable}
\end{table} 

\begin{figure}
\centering
\includegraphics[width=0.9\columnwidth]{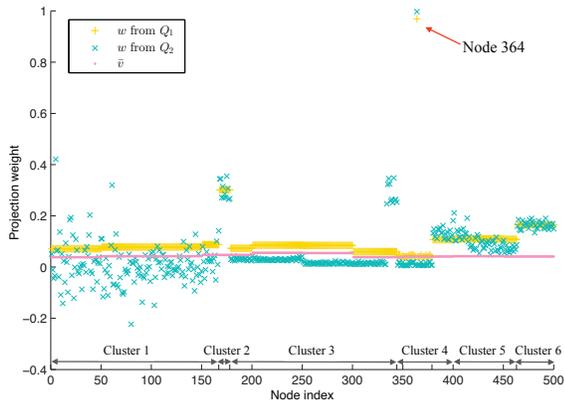}
        \caption{Weight designs from $Q_{1}$ and $Q_{2}$}
        \label{wplot}
\end{figure}
\begin{figure}
\centering
\includegraphics[width=0.9\columnwidth]{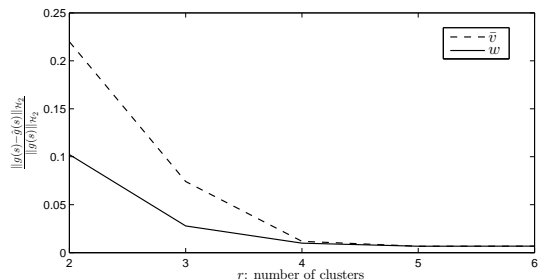}
        \caption{Performance loss by weight refinement}
        \label{iterplot}
\end{figure}

\subsection{Scalability Results}
To verify the scalability of Algorithm \ref{algk}, we increase the size of the network, and compare the computation time with that of solving a full-order LQR controller. We let $r=5$ and $\kappa=5$ for computing $\Phi_{\kappa}^{\frac{1}{2}}$, and present the results in Table \ref{tresult}. The table verifies the $\mathcal{O}(n^{3})$ complexity for full-order LQR vs. the $\mathcal{O}(n^{2}r)$ complexity for reduced-order LQR. When $n=8000$, the full-order LQR is already beyond the capability of our computation facility. The reduced-order LQR design, however, requires remarkably less computation time, while still providing a close performance match to the full-order LQR controller.

\begin{table}[h]
\caption{Scalability Results}
\label{tresult}
\centering
\begin{tabular}{ |c|c|c|c| }
\hline
\multirow{2}{*}{$n$} & \multicolumn{2}{|c|}{Computation time} & \multirow{2}{*}{Perf. loss $\frac{\| g(s)-\hat{g}(s) \|_{\mathcal{H}_{2}}}{\| g(s)\|_{\mathcal{H}_{2}}}$} \\
\cline{2-3}
& Full-order & Reduced-order & \\ \hline
1000 & $16.34$ sec & $0.79$ sec & $3.74 \%$ \\ \hline
2000 & $134.91$ sec & $2.64$ sec & $1.06 \%$ \\ \hline
4000 & $20.36$ min & $11.61$ sec & $0.25 \%$ \\ \hline
6000 & $71.06$ min & $24.73$ sec & $0.13 \%$\\ \hline
8000 & Out of memory & $48.39$ sec & - \\ \hline
10000 & Out of memory & $75.90$ sec & - \\ \hline
\end{tabular}
\end{table}

\section{Conclusion}
In this paper we developed a set of projection-based algorithms that improve the dynamic response of large-scale network systems with reduced-order LQR controllers. The advantage of these reduced-order controllers is that they are structured, and significantly easier to design and implement compared to regular full-order LQR controllers. Our future work will be to address the robustness of this approach to communication delays, to exploit additional input-output properties such as passivity to further improve performance, and to inspect the influence of network heterogeneity on clustering.

\section*{Appendix A: Special Case of Consensus Network}

The well-posedness conditions from Section \RNum{3}.A can be relaxed if system (\ref{full}) is a consensus network. Consensus is commonly used to model the dynamic behavior of many practical networks such as social networks, power networks, and wireless networks. For the same network graph $\mathcal{G} = (\mathcal{V},\mathcal{E})$ as in (\ref{full}), we suppose $n=n_{s}$. Each node has a real-valued node-weight $m_{i}>0$, and each edge $i j \in \mathcal{E}$ has a real-valued edge-weight $a_{ij} = a_{ji} > 0$. A consensus network can then be modeled in its simplest form by letting $A_{ij} = m_{i}^{-1} a_{ij}$ and $A_{ii} {=} - m_{i}^{-1} \sum_{j \in \mathcal{N}_{i}} a_{ij}$ in (\ref{nodesys}). The overall system writes as
\begin{align*}
\dot{x}_{M}(t) = M^{-1}L(\mathcal{G})x_{M}(t) + B_{M}u(t) + B_{dM}d(t),
\end{align*}
where $L(\mathcal{G})$ is the edge-weighted graph Laplacian matrix, and $M=diag([m_{1},...,m_{n}])$ is the matrix of node weights. For ease of analysis, let us consider a coordinate transformation $x = M^{\frac{1}{2}}x_{M}$, in which case the state matrices in (\ref{full}) become
\begin{align}
A = M^{-\frac{1}{2}}L(\mathcal{G})M^{-\frac{1}{2}}, B = M^{\frac{1}{2}}B_{M}, B_{d} = M^{\frac{1}{2}}B_{dM}. \label{consensus} 
\end{align}
Notice that $A$ here is a negative-semidefinite matrix, with only one zero eigenvalue at $A\bar{v}(A) = 0$, $\bar{v}(A) = \frac{1}{\sqrt{tr(M)}} M^{\frac{1}{2}} \mathbf{1}_{n}$, which we will simply denote as $\bar{v}$. The state matrix $\tilde{A}$ from the reduced-order system (\ref{reduced}) satisfies the following property.

\begin{custompro}{A.1}
$\tilde{A}$ is negative-semidefinite (or -definite) if $w \in  span(\bar{v})$ (or $w \not \in span(\bar{v})$).
\label{openpro}
\end{custompro}
\begin{proof}
Denote the complement of the projection matrix $P$ by $U$, such that $D=[P^{T}\ U^{T}]^{T}$ is unitary, i.e. $DD^{T} = I_{n}$. Then, $\tilde{A}=PAP^{T}$ is a leading principal of the matrix 
\begin{align*}
DAD^{T} = \begin{bmatrix}
PAP^{T} & PAU^{T} \\
UAP^{T} & UAU^{T}
\end{bmatrix},
\end{align*}
which is similar to $A$. Given the leading principal of a symmetric positive-definite (or-semidefinite) matrix still being positive-definite (or-semidefinite), it follows that $\tilde{A} \preceq 0$. Moreover, from Definition \ref{Pd}, $P$ is defined over $w$ such that $P^{T}Pw=w$. Then if $w \in  span(\bar{v})$, $\tilde{A}$ preserves the zero eigenvalue from $A$ since $\tilde{A}P\bar{v} = PA\bar{v} = 0$. Therefore, $\tilde{A}$ is negative-semidefinite if $w \in  span(\bar{v})$, or negative-definite otherwise.
\end{proof}

\subsubsection{Existence Condition}
The next theorem shows that the existence of $\tilde{X}$ for the reduced-order ARE (\ref{reare}) is guaranteed for any $w$ and $\mathcal{I}$, and matrices $\tilde{Q}=PQP^{T}$ and $\tilde{G}=PGP^{T}$.
\begin{customthm}{A.2}
If $(Q^{\frac{T}{2}},A)$ is detectable and $(A,G^{\frac{1}{2}})$ is stabilizable, then (\ref{reare}) is guaranteed with a unique stabilizing solution $\tilde{X} = \tilde{X}^{T} \succeq 0$.
\label{solvability}
\end{customthm}
\begin{proof}
For $Q\succeq 0$ and $G \succeq 0$, similar to $\tilde{A}$ in the proof of Proposition \ref{openpro}, we have both $\tilde{Q}\succeq 0$ and $\tilde{G} \succeq 0$. Recall that $(\tilde{Q}^{\frac{T}{2}},\tilde{A})$ is detectable if and only if for all $\lambda$ and $v$ that $\tilde{A}v = \lambda v$ and $\lambda \geq 0$, $\tilde{Q}^{\frac{T}{2}}v \neq 0$. If $w \in span(\bar{v})$, $\tilde{A}\preceq 0$ has only one zero eigenvalue, with the corresponding eigenvector $P\bar{v}$. Under this situation, $\bar{v}^{T}P^{T}\tilde{Q}P\bar{v} = \bar{v}^{T}Q\bar{v} \neq 0$ since $(Q^{\frac{T}{2}},A)$ is detectable. If $w \not \in span(\bar{v})$, $\tilde{A}$ is ￼￼negative-definite, which means $(\tilde{Q}^{\frac{T}{2}},\tilde{A})$ is trivially detectable. The same rational applies to showing stabilizability of $(\tilde{A},\tilde{G}^{\frac{1}{2}})$. Therefore, both stabilizability and detectability are satisfied, and thus (\ref{reare}) guarantees a unique positive-semidefinite solution.
\end{proof}

\subsubsection{Bound of ARE solution} 
The following lemma provides a specific value for $\beta(A,G,Q)$ in Lemma \ref{xblemma}.

\begin{customlem}{A.3}
The bound $\bar{\lambda}(X) \leq \beta(A,G,Q)$ holds for
\begin{align}
\beta(A,G,Q) = \frac{\bar{\sigma}(Q+G)}{2\underline{\sigma}(G-A)}.
\end{align}
\end{customlem} 
\begin{proof}
The expression for $\beta(A,G,Q)$ above can be simply obtained by letting $K_{t} = -G^{\frac{T}{2}}$, $D_{t} = I_{n}$ and $F = 2G - 2A \succ 0$ in Lemma \ref{xblemma}. 
\end{proof}
Note that the value specified by this Lemma only involves finding two extreme singular values, which can be computed very efficiently by Lanczos algorithm in $\mathcal{O}(n)$ complexity.

\subsubsection{Stability condition}
We next state two sufficient conditions that guarantee the stability of $\hat{g}(s)$ for consensus networks. 

\begin{customthm}{A.4}
The system $\hat{g}(s)$ is stable if $w=\bar{v}$, and $\mathcal{I}$ is an almost equitable partition \cite{equitable} of the graph $\mathcal{G}$, which means for $k\neq l$, the edge weight $a_{ij}$ is equal for all $i \in \mathcal{I}_{k}$ and $j \in \mathcal{I}_{l}$. 
\label{gstab1}
\end{customthm}
\begin{proof}
A similarity transformation of $D=[P^{T} \ U^{T}]^{T}$ and $D^{T}$ on $A-G\hat{X}$ yields
\begin{align}
D(A-G\hat{X})D^{T} = \begin{bmatrix}
\tilde{A} - \tilde{G}\tilde{X} & PAU^{T} \\
UAP^{T}-UGP^{T}\tilde{X} & UAU^{T}
\end{bmatrix}.
\end{align}
From \cite{equitable}, if $\mathcal{G}$ admits an almost equitable partition $\mathcal{I}$, the corresponding matrix $P$ with $w=\bar{v}$ will satisfy $A^{T}P^{T}=P^{T}N$ for some $N \in \mathbb{R}^{r \times r}$. As a result, $PAU^{T} = 0$, and thus $\hat{g}(s)$ is stable since $\tilde{A}-\tilde{G}\tilde{X}$ and $UAU^{T}$ are both Hurwitz given $w = \bar{v}$.
\end{proof}

\begin{customthm}{A.5}
Assume that $B$ is a square invertible matrix. Then, $\hat{g}(s)$ is stable if $G$ is similar to $\alpha I_{n}$ for some $\alpha >0$.
\label{gstab2}
\end{customthm}
\begin{proof}
Denote the right eigenspace of $G$ by $V$. If $G$ is similar to $\alpha I_{n}$, it follows that $V^{T}(G\hat{X}+\hat{X}G)V = 2 \alpha V^{T}\hat{X}V$. This means that $G\hat{X}+\hat{X}G$ is positive-semidefinite. From the matrix majorization property \cite{matrix}, we also have
\begin{align*}
2\bar{\lambda}(A-G\hat{X}) \leq z^{T}(A+A^{T}-G\hat{X}-\hat{X}G)z
\end{align*} 
hold for any non-zero vector $z$, with the RHS being non-positive given $A+A^{T} \preceq 0$ and $-G\hat{X}-\hat{X}G\preceq 0$. Hence for $\bar{\lambda}(A-G\hat{X})$ to be strictly negative, $z^{T}Az = 0$ should contradict $z^{T}G\hat{X}z = 0$. To prove the contradiction, we assume $z^{T}Az = 0$, namely $z\in span(\bar{v})$. Note that $z^{T}G\hat{X}z = 0$ holds if and only if either one of the following three conditions is satisfied: (1) $z\in ker(P)$, (2) $z \in ker(G)$, or (3) $Pz \in ker(\tilde{X})$. First of all, Assumption \ref{assw} requires $w_{\mathcal{I}_{i}}^{T}\bar{v}_{\mathcal{I}_{i}}$, $i=1,...,r$, which means $P\bar{v}\neq 0$ and thus $z\not \in ker(P)$. Given that $(A,BR^{-\frac{1}{2}})$ is stabilizable, we have $\bar{v}^{T}G\bar{v}\neq 0$, i.e. $z \not \in ker(G)$. To this end, the last condition remaining is $z\in span(\bar{v}) \not\Leftrightarrow Pz \in ker(\tilde{X})$ to complete the contradiction. By assuming a non-zero vector $v \in ker(\tilde{X})$, pre- and post-multiplying (\ref{reare}) with $v$ yields
\begin{align*}
v^{T}\tilde{A}^{T}\tilde{X}v + v^{T}\tilde{X}\tilde{A}v + v^{T}\tilde{Q}v - v^{T}\tilde{X}\tilde{G}\tilde{X}v=0.
\end{align*}
It can be easily verified that $ker(\tilde{X})$ is an $\tilde{A}$-invariant subspace contained in the null-space of $\tilde{Q}$. Given $z\in span(\bar{v})$, from the proof of Proposition \ref{openpro} we know that $Pz$ is $\tilde{A}$-invariant, i.e., $Pz$ is an eigenvector of $\tilde{A}$, when $w \in span(\bar{v})$. As a result, $z^{T}P^{T}\tilde{Q}Pz = z^{T}Qz \neq 0$ since $(Q^{\frac{T}{2}},A)$ is detectable. This verifies that $Pz \not \in ker(\tilde{Q})$, which proves $Pz \not \in ker(\tilde{X})$. Therefore, we conclude that $z^{T}Az = 0$ contradicts $z^{T}(G\hat{X}+\hat{X}G)z = 0$, and thus $\hat{g}(s)$ is stable. 
\end{proof}

\section*{Appendix B: Proofs}

\subsection{Proof of Theorem \ref{stab}}

Consider a Lyapunov function $V(x) = x^{T}Xx >0$, where $X\succ 0$ is the solution of ARE (\ref{are}). For $\hat{g}(s)$ to be asymptotically stable, $\dot{V}(x)$ needs to be negative, or equivalently
$$(A-G\hat{X})^{T}X+X(A-G\hat{X}) \prec 0.$$
Using the ARE (\ref{are}), the LMI above reduces to
$$Q \succ XGE + EGX - XGX,$$
which holds if $\underline{\lambda}(Q) > \bar{\lambda}(XGE + EGX - XGX)$. Notice that the RHS of the eigenvalue inequality follows
$$\bar{\lambda}(XGE + EGX - XGX) \leq \bar{\lambda}(XGE + EGX) + \bar{\lambda}(-XGX),$$
where we have respectively $\bar{\lambda}(XGE + EGX) \leq \bar{\sigma}(XGE + EGX) \leq 2\bar{\sigma}(XGE) \leq 2 \bar{\sigma}(X) \bar{\sigma}(G) \bar{\sigma}(E)$, and $-\bar{\lambda}(-XGX) = \underline{\sigma}(XGX) \geq \underline{\sigma}(X)^{2}\underline{\sigma}(G)$. Incorporating these two upper bounds yields the condition in (\ref{stabcond}). 


\subsection{Proof of Lemma \ref{vcond}}
From Lemma \ref{xblemma}, the ARE solution $X$ is bounded by $\bar{\sigma}(X) = \bar{\lambda}(X) \leq \beta(A,G,Q)$. Note that (\ref{stabcond}) is satisfied if
\begin{align}
\underline{\sigma}(Q)  >  2\bar{\sigma}(X) \bar{\sigma}(G)\bar{\sigma}(E),
\label{vcondp1}
\end{align}
where the RHS is further bounded by %
\begin{small}
\begin{align}
2\bar{\sigma}(X) \bar{\sigma}(G)\bar{\sigma}(E) \leq 2\beta(A,G,Q) \bar{\sigma}(G) [\beta(A,G,Q) {+} \bar{\sigma}(\tilde{X})]. \label{vcondp2}
\end{align}
\end{small}%
Therefore, (\ref{vcondp1}), and then (\ref{stabcond}) will hold if $\underline{\sigma}(Q)$ is greater than the RHS of (\ref{vcondp2}), which yields (\ref{vcondieq}) in Lemma \ref{vcond}. 

\subsection{Proof of Theorem \ref{tmain}}

The bound in (\ref{rboundn}) is derived assuming the worst case from (\ref{QRmod}), where $\tilde{Q} {=} PQP^{T} {+} \alpha I_{r}$ and $\tilde{G} {=} PQP^{T} {+} \alpha I_{r}$ for $\alpha > 0$. We divide the proof into three steps.
\subsubsection{}
We derive an analytical expression for $E$ by recovering the reduced-order ARE (\ref{reare}) to the full dimension as
\begin{align}
P^{T}(\tilde{A}^{T}\tilde{X} + \tilde{X}\tilde{A} + \tilde{Q} - \tilde{X}\tilde{G}\tilde{X})P = 0.\label{rare1}
\end{align}
Notice that $A$ and $\tilde{A}$ are related by
\begin{align}
\tilde{A}P = PA - PAU^{T}U,
\label{rare2}
\end{align}
where $U$ is the complement of $P$. Thereby substituting $\tilde{A}P$ and $P^{T}\tilde{A}^{T}$ in terms of (\ref{rare2}), and after a few calculations, (\ref{rare1}) yields the approximated ARE (for details, please see \cite{arepert})
\begin{align}
A^{T}\hat{X} + \hat{X}A + Q - \hat{X}G\hat{X} = \mathcal{R}, \label{area}
\end{align}
with the residue of the approximate ARE denoted by
\begin{align}
\mathcal{R} := \alpha\hat{X}^{2} + U^{T}UA^{T}\hat{X} + \hat{X}AU^{T}U + Q - P^{T}\tilde{Q}P. 
\label{residue}
\end{align}
By subtracting (\ref{are}) from (\ref{area}), we get the Sylvester equation 
\begin{align}
(A-G\hat{X})^{T}E + E(A-GX) = -\mathcal{R}. \label{sylv}
\end{align}
From (\ref{sylv}), we are able to explicitly write $E$ as a function of $A-GX$, $A-G\hat{X}$ and $\mathcal{R}$, and hence obtain an initial bound for $\| E\Phi^{\frac{1}{2}}\|_{F}$ in the next step.

\subsubsection{}
Pre- and post-multiplying (\ref{sylv}) with $\Phi^{\frac{T}{2}}$ and $\Phi^{\frac{1}{2}}$ respectively, the Sylvester equation takes the form
\begin{align}
\mathcal{A}_{1}\Phi^{\frac{T}{2}}E\Phi^{\frac{1}{2}} + \Phi^{\frac{T}{2}}E\Phi^{\frac{1}{2}}\mathcal{B}_{1} = -\Phi^{\frac{T}{2}} \mathcal{R}\Phi^{\frac{1}{2}}, 
\label{sylvw}
\end{align}
where we use the notations $\mathcal{A}_{1}=\Phi^{\frac{T}{2}}(A-G\hat{X})^{T}\Phi^{-\frac{T}{2}}$ and $\mathcal{B}_{1}=\Phi^{-\frac{1}{2}}(A-GX)\Phi^{\frac{1}{2}}$ for brevity. It can be easily shown that $\lambda(\mathcal{A}_{1}) < 0$ and $\lambda(\mathcal{B}_{1}) < 0$, which implies $\lambda_{i}(\mathcal{A}) + \lambda_{j}(\mathcal{B}) \neq 0$ for any $i,\ j=1,...,n$ so that (\ref{sylvw}) is solvable. Therefore, the weighted error $E\Phi^{\frac{1}{2}}$ can be expressed as
\begin{align}
E\Phi^{\frac{1}{2}} = \Phi^{-\frac{T}{2}} \cdot unvec[\mathcal{L}^{-1} \cdot vec(-\Phi^{\frac{T}{2}} \mathcal{R}\Phi^{\frac{1}{2}})],
\label{sylvo}
\end{align}
where $\mathcal{L} = I_{n} \otimes \mathcal{A}  + \mathcal{B}^{T} \otimes I_{n}$ is an $n^{2}\times n^{2}$ matrix. Since the Frobenius norm is unitary invariant, taking norm on both sides of (\ref{sylvo}) provides an upper bound on $\| E\Phi^{\frac{1}{2}}\|_{F}$ as
\begin{align}
\| E\Phi^{\frac{1}{2}}\|_{F} \leq \bar{\sigma}( \mathcal{L}^{-1}) \bar{\sigma}( \Phi^{-\frac{1}{2}} ) \| \Phi^{\frac{T}{2}} \mathcal{R} \Phi^{\frac{1}{2}} \|_{F}.
\label{t1}
\end{align}
Note that $\bar{\sigma}( \mathcal{L}^{-1})$ follows $\bar{\sigma}(\mathcal{L}^{-1}) = \frac{1}{\underline{\sigma}(\mathcal{L})}$, where $\underline{\sigma}(\mathcal{L})$ is calculated by $\underline{\sigma}^{2}(\mathcal{L}) = \underline{\lambda}(\mathcal{L}\mathcal{L}^{T})$ with
\begin{align*}
\mathcal{L}\mathcal{L}^{T} = I_{n} \otimes \mathcal{A}\mathcal{A}^{T} + \mathcal{B}^{T}\mathcal{B} \otimes I_{n} + \mathcal{B}^{T} \otimes \mathcal{A}^{T} + \mathcal{B} \otimes \mathcal{A}.
\end{align*}
The eigenvalues of $\mathcal{B} \otimes \mathcal{A}$ are counted by $\lambda_{i}(\mathcal{A})\lambda_{j}(\mathcal{B})$ with $i,\ j=1,...n$, and according to the Weyl's inequality of eigenvalues \cite{matrix}, we have the lower bound for $\underline{\sigma}^{2}(\mathcal{L})$ as
\begin{align}
\nonumber
\underline{\sigma}^{2}(\mathcal{L}) & \geq \underline{\lambda}(\mathcal{A}\mathcal{A}^{T}) + \underline{\lambda}(\mathcal{B}^{T}\mathcal{B}) + 2\underline{\lambda}(\mathcal{B} \otimes \mathcal{A}) \\
& = \underline{\sigma}^{2}(\mathcal{A}) + \underline{\sigma}^{2}(\mathcal{B}) + 2\bar{\lambda}(\mathcal{A})\bar{\lambda}(\mathcal{B}) \geq \underline{\sigma}^{2}(\mathcal{B}). \label{t2}
\end{align}
Combining (\ref{t1}) with (\ref{t2}) then yields the following bound
\begin{align}
\| E\Phi^{\frac{1}{2}}\|_{F} \leq \epsilon_{1} \| \Phi^{\frac{T}{2}} \mathcal{R} \Phi^{\frac{1}{2}} \|_{F},
\label{thb}
\end{align}
where $\epsilon_{1} = \frac{\bar{\sigma}( \Phi^{-\frac{1}{2}} ) }{\underline{\sigma}[\Phi^{-\frac{1}{2}}(A-GX)\Phi^{\frac{1}{2}}]} >0$ is independent of $P$. In (\ref{thb}), the norm of the weighted residue $\Phi^{\frac{T}{2}}\mathcal{R}\Phi^{\frac{1}{2}}$, written by%
\begin{small}
\begin{align}
\nonumber
& \Phi^{\frac{T}{2}} \mathcal{R} \Phi^{\frac{1}{2}} = \alpha \Phi^{\frac{T}{2}} \hat{X}^{2} \Phi^{\frac{1}{2}} + \Phi^{\frac{T}{2}} U^{T}UA^{T}\hat{X} \Phi^{\frac{1}{2}} + \Phi^{\frac{T}{2}} \hat{X}AU^{T}U \Phi^{\frac{1}{2}} \\
& \quad \ \ + \Phi^{\frac{T}{2}}(QU^{T}U  + U^{T}UQ - U^{T}UQU^{T}U - \alpha P^{T}P)\Phi^{\frac{1}{2}}, \label{Re}
\end{align} 
\end{small}%
contains the inexplicit functional $\hat{X}$. We then bypass this term in the final step. 

\subsubsection{}
Taking norm on both sides of (\ref{Re}), and then isolating the norm of $\hat{X}$, we can form the bound 
\begin{align}
\nonumber
 \| \Phi^{\frac{T}{2}}\mathcal{R}\Phi^{\frac{1}{2}} \|_{F}  \leq  & 2 [\bar{\sigma}(A) \bar{\sigma}(\Phi^{\frac{1}{2}}) \bar{\sigma}(\hat{X})  + \bar{\sigma}(Q\Phi^{\frac{1}{2}}) ] \xi   \\
& + \bar{\sigma}(Q) \xi^{2} +  \alpha \bar{\sigma}(\Phi) [\bar{\sigma}^{2}(\hat{X}) + 1 ], \label{rbound}
\end{align}
with $\xi =\| U^{T}U\Phi^{\frac{1}{2}} \|_{F}$. Recall that $\bar{\sigma}(\hat{X}) = \bar{\sigma}(P^{T}\tilde{X}P) = \bar{\sigma}(\tilde{X})$, where $\tilde{X}$ is the solution of the reduced-order ARE (\ref{reare}). The norm $\bar{\sigma}(\tilde{X})$ can be further bounded by $\beta(\tilde{A},\tilde{G},\tilde{Q})$ through Lemma \ref{xblemma}. Theorem \ref{tmain}, therefore, follows from (\ref{thb}), (\ref{rbound}) and $\bar{\sigma}(\tilde{X}) \leq \beta(\tilde{A},\tilde{G},\tilde{Q}) \leq \mathrm{sup}_{P} \beta(\tilde{A},\tilde{G},\tilde{Q})$.

\subsection{Proof of Lemma \ref{kerror}}
To prove the error bound (\ref{xierror}), we define a matrix $\bar{\Phi}$ as
\begin{align}
\bar{\Phi} =\begin{bmatrix} \bar{\Phi}_{1} & \cdots & \bar{\Phi}_{n_{b}}
\end{bmatrix},
\label{nphi}
\end{align}
where $\bar{\Phi}_{i}=Ydiag(Y^{-1}b_{i})\mathcal{C}^{\frac{1}{2}},\ i=1,...,n_{b}$ and $b_{i}$ is the $i^{th}$ column of $B_{d}$. The matrix $\bar{\Phi}$ satisfies $\Phi = \bar{\Phi}\bar{\Phi}^{T} = \Phi^{\frac{1}{2}}\Phi^{\frac{T}{2}}$. Besides $Y_{1}$ and $\Omega_{1}$ defined in Definition \ref{Phik}, we further denote $Y_{2}= Y_{:,\kappa+1,n}$ and $\Omega_{2}=Y^{-1}_{\kappa+1:n,:}$, and partition the Cholesky Decomposition $\mathcal{C}^{\frac{1}{2}}$ as %
\begin{align*}
\mathcal{C}^{\frac{1}{2}}  = \begin{bmatrix}
\mathcal{C}^{\frac{1}{2}}_{1,1} & 0 \\
\mathcal{C}^{\frac{1}{2}}_{2,1} & \mathcal{C}^{\frac{1}{2}}_{2,2}
\end{bmatrix} = \begin{bmatrix}
\mathcal{C}^{\frac{1}{2}}_{1{:}\kappa,1{:}\kappa} & 0 \\
\mathcal{C}^{\frac{1}{2}}_{\kappa{+}1{:}n,1{:}\kappa} & \mathcal{C}^{\frac{1}{2}}_{\kappa{+}1{:}n,\kappa{+}1{:}n}
\end{bmatrix}.
\end{align*}
With these notations, $\bar{\Phi}_{i}$ in (\ref{nphi}) can be decomposed into $ \bar{\Phi}_{i} = \bar{\Phi}_{i,s} + \bar{\Phi}_{i,f} $, where $\bar{\Phi}_{i,s} {=} \begin{bmatrix}
Y_{1}diag(\Omega_{1}b_{i})\mathcal{C}^{\frac{1}{2}}_{1,1} & 0
\end{bmatrix}$ and $\bar{\Phi}_{i,f} = \begin{bmatrix}
Y_{2}diag(\Omega_{2}b_{i})\mathcal{C}^{\frac{1}{2}}_{2,1} & Y_{2}diag(\Omega_{2}b_{i})\mathcal{C}^{\frac{1}{2}}_{2,2}
\end{bmatrix}$, and thus $\Phi_{\kappa}$ can be rewritten as $\Phi_{\kappa} {=} \sum_{i=1}^{n_{b}} \bar{\Phi}_{i,s}\bar{\Phi}_{i,s}^{T}$. Notice that $\xi = \| (I_{n}-P^{T}P) \Phi^{\frac{1}{2}}\|_{F} = \| (I_{n}-P^{T}P) \bar{\Phi}\|_{F}$ satisfies 
\begin{align}
\xi_{\kappa}^{*} \leq \xi \leq \| (I_{n}-P^{T}P)\Phi_{\kappa}^{\frac{1}{2}} \|_{F} + \| (I_{n}-P^{T}P) \bar{\Phi}_{f} \|_{F},
\label{xi2}
\end{align}
where $\bar{\Phi}_{f} = \begin{bmatrix}
\bar{\Phi}_{1,f} & \cdots & \bar{\Phi}_{n_{b},f}
\end{bmatrix}$. 
The second norm on the RHS of (\ref{xi2}) is further bounded by $\| (I_{n}-P^{T}P)\bar{\Phi}_{f} \|_{F} \leq \| \bar{\Phi}_{f}\|_{F}$ with %
\begin{small}
\begin{align*}
 \| \bar{\Phi}_{f}\|_{F} & \leq \sqrt{\eta^{2} \sum_{i=1}^{n_{b}} (\| \mathcal{C}^{\frac{1}{2}}_{2,1}\|_{F}^{2} + \| \mathcal{C}^{\frac{1}{2}}_{2,2}\|_{F}^{2}) } = \sqrt{\eta^{2} n_{b}\sum_{i=\kappa+1}^{n} -\frac{1}{2\lambda^{-}_{i}}}.
\end{align*} %
\end{small}%
Inserting this along with $P = \mathrm{argmin}_{P} \ \xi_{\kappa}$ to the RHS of (\ref{xi2}) yields the error bound (\ref{xierror}). 

\subsection{Proof of Lemma \ref{circle}} 

Denote $n_{i}=|\mathcal{I}_{i}|_{c},\ i=1,...,r$, the objective function $f_{\bar{V}}$ in (\ref{opt2}) can be expanded as 
\begin{align*}
& f_{\bar{V}} = tr(\bar{V}^{T}P^{T}PQP^{T}P\bar{V}) = (w\ast w)^{T} (Q\ast \bar{V}\bar{V}^{T}) (w\ast w)
\end{align*}
where $\ast$ is the Khatri-Rao product defined by %
\begin{small}
\begin{align*}
Q\ast \bar{V}\bar{V}^{T} {=} \begin{bmatrix}
Q_{\mathcal{I}_{1},\mathcal{I}_{1}}\otimes \bar{V}_{\mathcal{I}_{1},:}\bar{V}_{\mathcal{I}_{1},:}^{T} & \cdots & Q_{\mathcal{I}_{1},\mathcal{I}_{r}}\otimes \bar{V}_{\mathcal{I}_{1},:}\bar{V}_{\mathcal{I}_{r},:}^{T} \\
\vdots & \ddots & \vdots \\
Q_{\mathcal{I}_{r},\mathcal{I}_{1}}\otimes \bar{V}_{\mathcal{I}_{r},:}\bar{V}_{\mathcal{I}_{1},:}^{T} & \cdots & Q_{\mathcal{I}_{r},\mathcal{I}_{r}}\otimes \bar{V}_{\mathcal{I}_{r},:}\bar{V}_{\mathcal{I}_{r},:}^{T}
\end{bmatrix},
\end{align*}
\end{small} %
\begin{align*}
w\ast w = \begin{bmatrix}
w_{\mathcal{I}_{1}}^{T} \otimes w_{\mathcal{I}_{1}}^{T} & \cdots & w_{\mathcal{I}_{r}}^{T} \otimes w_{\mathcal{I}_{r}}^{T}
\end{bmatrix}^{T}.
\end{align*}
Denote the block-diagonal submatrix of $Q$ by $Q_{d}$, i.e. $Q_{d} = diag(Q_{\mathcal{I}_{1},\mathcal{I}_{1}},\cdots, Q_{\mathcal{I}_{r},\mathcal{I}_{r}})$, and the off-diagonal by $Q_{o} = Q-Q_{d}$. Therefore, we can find the difference between objective functions of (\ref{opta}) and (\ref{optf}) as $(w\ast w)^{T}(\rho Q_{o}\ast \bar{V}\bar{V}^{T})(w\ast w)$. According to the Gershgorin circle theorem, the eigenvalues of $\rho Q_{o}\ast \bar{V}\bar{V}^{T}$ are all bounded inside the range of $(-\| \rho Q_{o}\ast \bar{V}\bar{V}^{T}\|_{1},\| \rho Q_{o}\ast \bar{V}\bar{V}^{T}\|_{1})$. Therefore, the theorem follows when $J_{e} = \| \rho Q_{o}\ast \bar{V}\bar{V}^{T}\|_{1}$.

\renewenvironment{IEEEbiography}[1]
  {\IEEEbiographynophoto{#1}}
  {\endIEEEbiographynophoto}

\begin{biography}{Nan Xue}
(S'15) received his B.E. degree in Electrical Engineering from Xi'an Jiaotong University, China in 2013. He is currently pursuing his PhD degree in Electrical Engineering at North Carolina State University, Raleigh, NC. His research interests include analysis, control and model reduction of large-scale networked dynamic systems and power systems.
\end{biography}

\begin{biography}{Aranya Chakrabortty}
(M'08, SM'15) received his PhD degree in Electrical Engineering from Rensselaer Polytechnic Institute, Troy, NY in 2008. From 2008 to 2009 he was a postdoctoral research associate at the Aeronautics and Astronautics department of University of Washington, Seattle. He is currently an Associate Professor in the Electrical and Computer Engineering department of North Carolina State University, Raleigh, NC, where he is also affiliated to the FREEDM Systems Center. His research interests are in all branches of control theory with applications to power systems, especially in wide-area monitoring and control of large power systems using Synchrophasors. He received the NSF CAREER award in 2011.
\end{biography}

\end{document}